\documentclass[11pt]{article}
\usepackage{palatino}
\usepackage{microtype}
\usepackage{pgfplots}
\usepackage{tikz, tikzpeople}
\usetikzlibrary{arrows}
\usetikzlibrary{shapes.multipart}
\usetikzlibrary{fit}
\usetikzlibrary{shapes.geometric,positioning}
\usetikzlibrary{decorations.pathreplacing}
\usetikzlibrary{calc}
\usepackage{bm}
\usepackage[utf8]{inputenc}
\usepackage{mathtools}
\usepackage{graphicx}
\usepackage{algorithm}
\usepackage{cite}
\usepackage{amsmath,amssymb,amsfonts, amsthm,  mathrsfs}
\usepackage{algpseudocode}
\usepackage{textcomp}
\usepackage{xcolor}
\usepackage{enumitem}
\usepackage{cancel}
\usepackage{hyperref}

\usetikzlibrary{patterns}
\newtheorem{lemma}{Lemma}
\newtheorem{theorem}{Theorem}
\newtheorem{example}{Example}
\newtheorem{remark}{Remark}
\newtheorem{corollary}{Corollary}

\newlength{\leftstackrelawd}
\newlength{\leftstackrelbwd}
\def\leftstackrel#1#2{\settowidth{\leftstackrelawd}%
{${{}^{#1}}$}\settowidth{\leftstackrelbwd}{$#2$}%
\addtolength{\leftstackrelawd}{-\leftstackrelbwd}%
\leavevmode\ifthenelse{\lengthtest{\leftstackrelawd>0pt}}%
{\kern-.5\leftstackrelawd}{}\mathrel{\mathop{#2}\limits^{#1}}}

\newcommand{\PH}{\mathsf{str}}
\newcommand{\mN}{\mathsf{N}} 
\newcommand{\mZ}{\mathsf{Z}} 
\newcommand{\mP}{\mathsf{P}} 
\newcommand{\mE}{\mathsf{E}} 
\newcommand{\mD}{\mathsf{D}} 
\newcommand{\mB}{\mathsf{B}} 
\newcommand{\mQ}{\mathsf{Q}} 

\newcommand{\C}{\mathchoice
  {\mathrm{\scriptscriptstyle C}} 
  {\mathrm{\scriptscriptstyle C}} 
  {\mathrm{\scriptscriptstyle C}} 
  {\mathrm{\scriptscriptstyle C}} 
}
\newcommand{\NS}{\mathchoice
  {\mathrm{\scriptscriptstyle NS}} 
  {\mathrm{\scriptscriptstyle NS}} 
  {\mathrm{\scriptscriptstyle NS}} 
  {\mathrm{\scriptscriptstyle NS}} 
}

\newcommand{\NSa}{\mathchoice
  {\mathrm{\scriptscriptstyle NS\text{-}1}} 
  {\mathrm{\scriptscriptstyle NS\text{-}1}} 
  {\mathrm{\scriptscriptstyle NS\text{-}1}} 
  {\mathrm{\scriptscriptstyle NS\text{-}1}} 
}

\newcommand{\NSz}{\mathchoice
  {\mathrm{\scriptscriptstyle NS\text{-}0}} 
  {\mathrm{\scriptscriptstyle NS\text{-}0}} 
  {\mathrm{\scriptscriptstyle NS\text{-}0}} 
  {\mathrm{\scriptscriptstyle NS\text{-}0}} 
}

\newcommand{\SIa}{\mathchoice
  {\mathrm{\scriptscriptstyle SI\text{-}1}} 
  {\mathrm{\scriptscriptstyle SI\text{-}1}} 
  {\mathrm{\scriptscriptstyle SI\text{-}1}} 
  {\mathrm{\scriptscriptstyle SI\text{-}1}} 
}

\newcommand{\semidet}{\mathchoice
  {\mathrm{\scriptscriptstyle  SD}} 
  {\mathrm{\scriptscriptstyle  SD}} 
  {\mathrm{\scriptscriptstyle  SD}} 
  {\mathrm{\scriptscriptstyle  SD}} 
}

\newcommand{\semidetE}{\mathchoice
  {\mathrm{\scriptscriptstyle  SDE}} 
  {\mathrm{\scriptscriptstyle  SDE}} 
  {\mathrm{\scriptscriptstyle  SDE}} 
  {\mathrm{\scriptscriptstyle  SDE}} 
}

\newcommand{\KS}{\mathchoice
  {\mathrm{\scriptscriptstyle  KS}} 
  {\mathrm{\scriptscriptstyle  KS}} 
  {\mathrm{\scriptscriptstyle  KS}} 
  {\mathrm{\scriptscriptstyle  KS}} 
}

\newcommand{\Sato}{\mathchoice
  {\mathrm{\scriptscriptstyle  Sato}} 
  {\mathrm{\scriptscriptstyle  Sato}} 
  {\mathrm{\scriptscriptstyle  Sato}} 
  {\mathrm{\scriptscriptstyle  Sato}} 
}

\newcommand{\tR}{\mathchoice
  {\mathrm{\scriptscriptstyle  R}} 
  {\mathrm{\scriptscriptstyle  R}} 
  {\mathrm{\scriptscriptstyle  R}} 
  {\mathrm{\scriptscriptstyle  R}} 
}

\newcommand{\tT}{\mathchoice
  {\mathrm{\scriptscriptstyle  T}} 
  {\mathrm{\scriptscriptstyle  T}} 
  {\mathrm{\scriptscriptstyle  T}} 
  {\mathrm{\scriptscriptstyle  T}} 
}

\newcommand{\mathrmtiny}[1]{\text{\tiny$\mathrm{#1}$}}

\allowdisplaybreaks

\title{Virtual Signaling of CSIT via Non-Signaling Assistance}

\author{Yuhang Yao, Syed A. Jafar\\
{\small Center for Pervasive Communications and Computing (CPCC)}\\
{\small University of California Irvine, Irvine, CA 92697}\\
{\small \it Email: \{yuhangy5, syed\}@uci.edu}
}

\date{}

\begin{document}
\maketitle

\begin{abstract}
Non-signaling correlations, which (strictly) include  quantum correlations, provide a tractable path to explore the potential impact of quantum nonlocality on the capacity of classical communication networks. Motivated by a recent discovery that certain  wireless  network settings benefit significantly from non-signaling (NS) correlations, various generalizations are considered. First, it is shown that for a point to point discrete memoryless channel $\mN_{Y|XS}$ with input $X$, output $Y$, channel state $S$, and non-causal channel state information at the transmitter (CSIT), the NS-assisted Shannon capacity is $\max_{\mP_{X|S}}I(X;Y|S)$, matching the classical (without NS assistance) capacity of the channel for the setting where $S$ is also made available to the receiver. For any finite blocklength, the NS-assisted optimal probability of successful decoding over $\mN_{Y|XS}$  is shown to remain unchanged if in addition $S$ is made available to the receiver. The key insight is summarized as `virtual signaling of CSIT via NS-assistance' and is supported by further results as follows. For a discrete memoryless $2$-user broadcast channel (BC), the Shannon capacity region with NS-assistance available only between the transmitter and User $1$, is found next. Consistent with the aforementioned key insight, the result matches the classical capacity region for the setting where the desired message of User $2$ is made available in advance as side-information to User $1$. The latter capacity region is known from a result of Kramer and Shamai. Next, for a semi-deterministic BC, the Shannon capacity region with full (tripartite) NS-assistance is shown to be the same as if only bipartite NS-assistance was available between the transmitter and the non-deterministic user. Bipartite NS-assistance between the transmitter and only the deterministic user, does not improve the capacity region relative to the corresponding classical setting. The capacity region is also found for the BC obtained by passing the outputs of a semi-deterministic BC through separate erasure channels, provided that the deterministic output does not experience the worse erasure channel. The result matches Sato's converse bound. Finally, the analysis is extended to a $K$-user BC with full NS-assistance among all parties. It is shown that the  optimal probability of successful decoding for any finite blocklength does not change if each User $k, k\in[K]$ is provided in advance the desired messages of Users $k+1, k+2,\cdots, K$.

\end{abstract}

\newpage
\section{Introduction}
The \emph{no-signaling principle} is a fundamental consequence of special relativity that precludes  any controllable physical influence from propagating beyond the light cone, and therefore forbids superluminal (faster than light) communication. It plays a central role in reconciling \emph{quantum non-locality} --- non-local correlation  across space-like separated quantum systems --- with \emph{relativistic causality}. Indeed, quantum non-locality is a non-signaling (NS)  resource,\footnote{It is known that the set of NS resources, i.e., the set of correlations that do not violate the no-signaling principle, includes all quantum correlations, and the inclusion is strict.} and  \emph{by itself}  such a resource shared across multiple parties, cannot allow any communication among those parties.

This leads naturally to the following questions: What if the parties sharing a NS resource are also connected by a classical communication network? Can NS assistance improve the Shannon capacity region of that classical communication network? How significant can such an improvement be? What kinds of communication networks allow significant improvements? Are such networks \emph{naturally} encountered? And, how much of these improvements are achievable with quantum correlations? These  fundamental  questions lie at the intersection of classical and quantum information theory and define a research frontier that has experienced  recent progress \cite{fawzi2024MAC, pereg2021quantumBC, fawzi2024broadcast,Yao_Jafar_NS_DoF}. The motivation of  this work is  to further advance this frontier.

While not all NS correlations are realizable by quantum systems, the simpler formulation of NS correlations  makes them more amenable to analysis than quantum correlations.  Thus, NS models offer, firstly, a tractable path to explore the potential impact of \emph{quantum}-nonlocality on the capacity of classical communication networks. Clearly, the quantum-assisted capacity region of a classical communication network is sandwiched between its capacity regions with and without NS-assistance. It follows then that strong benefits of quantum nonlocality can only exist where large  improvements are possible through NS-assistance. In this sense, the search for large improvements in capacity due to NS-assistance, serves a similar purpose as a \emph{metal detector for a treasure hunt}. Secondly, NS models, which encapsulate the maximal non-locality that is allowed by relativistic causality, are especially important within the broad framework of \emph{generalized probabilistic theories} \cite{GPT2007, GPT2023} to explore the fundamental limits of information processing supported by non-local correlations, and thereby potentially discover new information processing principles as well as new physical theories.
Thirdly, the study of NS-assisted capacity  contributes new converse bounds, or a better understanding of existing converse bounds in classical information theory. Examples include the  Polyanskiy-Poor-Verdu one-shot metaconverse \cite{Polyanskiy_Poor_Verdu}, a cornerstone of finite-blocklength classical information theory, which is alternatively recovered (\emph{and found to be tight!}) under NS assistance \cite{matthews2012linear}, and  algorithmic approaches to classical converses for one-shot settings inspired by NS-assisted models \cite{Barman_Fawzi, Jose_Kulkarni_2019, fawzi2024broadcast}. Finally, NS-assistance is also explored in stochastic control theory for decentralized information structures \cite{Venkat_NS_2007,Jose_Kulkarni_2015,  Deshpande_2023, Dhingra_Kulkarni_2024}, with several intriguing parallels, e.g., between \cite{Jose_Kulkarni_2015, matthews2012linear, Polyanskiy_Poor_Verdu}.

Many of the results on NS-assisted capacity of classical communication networks are negative results, showing that there can be no Shannon capacity advantage from NS  assistance (which subsumes quantum entanglement assistance) in certain settings. It is known, for instance,  that NS assistance does not improve the Shannon capacity of a classical point to point channel\cite{matthews2012linear, Barman_Fawzi}, generalizing a prior result for quantum assistance from \cite{Bennett_Shor_Smolin_Thapliyal_PRL}. NS assistance also does not improve the capacity region of a deterministic broadcast channel (BC) \cite{ fawzi2024broadcast}, or a BC where NS resources are only shared among receivers \cite{fawzi2024broadcast} (generalizing a result for the corresponding quantum-assisted setting from \cite{pereg2021quantumBC}), or a multiple access channel (MAC) where independent NS resources are shared between each transmitter and the receiver \cite{fawzi2024MAC}. Positive results, showing strict capacity \emph{improvements} from NS assistance started with the work of Quek and Shor \cite{Quek_Shor} on interference channels, which introduced the key idea of  translating NS and/or quantum strategies for nonlocal games into capacity improvements for carefully constructed artificial channels. This was followed by the works by Leditzky et al. \cite{leditzky2020playing}, Seshadri et al. \cite{seshadri2023separation}, Fawzi and Ferme \cite{fawzi2024MAC}, and Pereg et al. \cite{pereg2024MAC_QEassist} on multiple access  channels where similar advantages were established for NS and/or quantum-assistance, and by Hawellek et al. who established capacity advantages in interference channels with entangled transmitters \cite{hawellek2024interferencechannelentangledtransmitters}. The capacity of NS-assisted broadcast channels was studied by Fawzi and Ferme in \cite{fawzi2024broadcast}, who posed the question: can NS assistance  improve the Shannon capacity region of a broadcast channel? The question was answered in the affirmative in \cite{Yao_Jafar_NS_DoF}, somewhat surprisingly even for semi-deterministic and degraded broadcast channels where an improvement was not expected, and it was shown that for certain $K$ user BC's that occur naturally in wireless networks, the  improvement in Shannon capacity due to NS assistance approaches a factor of $K$. On one hand, considering that the improvements in Shannon capacity reported previously were relatively modest amounts (not exceeding 5\% \cite{seshadri2023separation,hawellek2024interferencechannelentangledtransmitters} to our knowledge\footnote{Note that we are considering Shannon capacity (requiring vanishing error for large block lengths) of discrete memoryless channels. Larger improvements from NS assistance are  indeed well-known in other settings, e.g., zero-error capacity \cite{cubitt2011zero,cubitt2010improving}, arbitrarily varying channel capacity \cite{Notzel}, and communication cost for computations such as binary decision problems \cite{van2013implausible}. Quek and Shor \cite{Quek_Shor} demonstrated a factor of $2$ improvement in a $2$ user IC, but their notion of  sum-rate `capacity' which corresponds to maximizing single-letter mutual informations for the two users, $\max I(X_1;Y_1)+I(X_2;Y_2)$, is different from Shannon capacity.}), a $K$-fold improvement in Shannon capacity seems shockingly large (e.g., 100\% improvement for $K=2$ users, 900\% for $K=10$ users) and represents a strong `metal detector' signal in our earlier analogy. On the other hand it is important to note that it is not yet known how much, if any, of this gain is actually achievable under the restriction to strictly quantum-correlations. The scenarios pinpointed in \cite{Yao_Jafar_NS_DoF} provide a focused target for future quantum-theoretic analysis. For our present purpose let us note that reference \cite{Yao_Jafar_NS_DoF} also considers a NS-assisted single-user \emph{fading-dirt} channel model,  where even larger (\emph{unbounded} as channel alphabet size approaches infinity) \emph{multiplicative} gains in Shannon capacity are established. The present work  seeks a more fundamental understanding of these findings.

There are three classes of communication networks considered in this work. The first is a general point to point discrete memoryless channel $\mN_{Y|XS}$ with input $X$, output $Y$, and state $S$ that is known non-causally to the transmitter. By the well known Gelfand-Pinsker Theorem \cite{gel1980coding} the classical capacity of such a channel is  $C^{\C}(\mN) = \max_{\mP_{X U}} I(U;Y) - I(U;S)$. We show (Theorem \ref{thm:capacity_cws}) that the NS-assisted capacity  of this channel is $C^{\NS}(\mN) = \max_{\mP_{X|S}} I(X;Y|S)$. The form of the capacity expression is recognizable as the classical capacity of a variant of the channel setting $\mN$, denoted $\bar\mN$, where the state $S$ is also made available to the \emph{receiver} \cite[Sec. 7.4.1]{NIT}. Thus, $C^{\NS}(\mN) = C^{\C}(\bar\mN)$. It is noteworthy that the multiplicative gap between $C^{\C}(\mN)$ and $C^{\NS}(\mN) = C^{\C}(\bar\mN)$ is in general \emph{unbounded} (Example \ref{ex:fadedirt}), underpinning the corresponding observation in \cite{Yao_Jafar_NS_DoF}. Going beyond Shannon capacity, we study optimal probability of successful decoding (denoted as $\eta$) over finite blocklengths. We prove (Theorem \ref{thm:csit_tp}) that with NS-assistance, for any finite-blocklength, the  optimal probability of success over $\mN$ is the same as over $\bar\mN$, i.e., $\eta(\mN)=\eta(\bar\mN)$. In words, with NS-assistance, the channel setting $\mN$ where the state $S$ is known to the transmitter, is exactly as good as the channel setting $\bar\mN$ where the state $S$ is also given to the receiver. Equivalence of finite-blocklength success probability is stronger than, and therefore implies, the equivalence of Shannon capacity, i.e., $C^{\NS}(\mN) = C^{\NS}(\bar\mN)$. In \emph{effect},  it is as if NS-assistance allows CSIT to be shared with the receiver. This insight, which is encapsulated as the titular `\emph{virtual signaling of CSIT via NS-assistance},' is the central insight of this work. 
The intuition of `virtual CSIT signaling' is  helpful to anticipate and interpret our results for the next two classes of communication networks.

The second class that we consider comprises $2$ user discrete memoryless broadcast channels, $\mN_{Y_1Y_2|X}$. While in this case there is no notion of state inherent to the channel \emph{per se}, it is well known that broadcast channels are intimately related to channels with state known non-causally at the transmitter. This is because the message (say, $W_2$) to be transmitted to a user (User $2$) impacts the codebook that can be used to communicate with the other user (User $1$). Since $W_2$ is known non-causally (i.e., prior to the beginning of transmission) to the transmitter, it acts somewhat like a channel state that is known non-causally to the transmitter, from the perspective of the point to point communication between the transmitter and User $1$. For a general $2$ user discrete memoryless broadcast channel $\mN_{Y_1Y_2|X}$ with NS-assistance between the transmitter and User $1$, we prove (Theorem \ref{thm:bipartite_NSBC}) that the capacity region $\mathcal{C}^{\NSa}(\mN)$ is the union (over all $\mP_{XU}$) of the rate pairs $(R_1,R_2)$ such that $R_1\leq I(X;Y_1)$, $R_2\leq I(U;Y_2)$ and $R_1+R_2\leq I(X;Y_1\mid U)+I(U;Y_2)$. This region is recognizable as the classical capacity region $\mathcal{C}^{\C,\SIa}(\mN)$ of a related setting of $\mN$, where User $2$'s desired message $W_2$ is provided in advance to User $1$ as side-information. The characterization of $\mathcal{C}^{\C,\SIa}(\mN)$ is available explicitly as a special case of a more general result of Kramer and Shamai in \cite{kramer2007capacity}.
In particular, we show (Theorem \ref{thm:bipartite_NSBC}) that $\mathcal{C}^{\NSa}(\mN)=\mathcal{C}^{\NS,\SIa}(\mN)=\mathcal{C}^{\C,\SIa}(\mN)$, i.e., NS assistance between transmitter and User $1$ is as helpful as providing User $1$ with $W_2$, but not any more helpful than that. Note that the result is consistent with the `virtual signaling of CSIT' insight, as in this case, it is as if NS-assistance to User $1$ effectively signaled the CSIT $(W_2)$ to User $1$. Within the class of $2$-user broadcast channels, we then consider the sub-class of \emph{semi-deterministic} broadcast channels. For this class we characterize the capacity region even with full (tripartite) NS-assistance between the transmitter and both receivers. The capacity region is the same as if NS-assistance is provided only between the transmitter and the non-deterministic receiver. In fact if NS-assistance is provided only between the transmitter and the \emph{deterministic} receiver, then it does not improve the capacity region at all compared to the classical setting (without NS assistance). This is consistent with the observation that in a semi-deterministic BC,  the capacity region is unchanged if the deterministic receiver is provided the desired message of the other user as side information.  We then offer a limited generalization of the semi-deterministic BC, to the stochastic BC obtained by passing the semi-deterministic BC outputs through separate erasure channels for the two users (provided that the deterministic output does not experience the worse erasure channel). The capacity region with full (tripartite) NS assistance is characterized for this setting in Corollary \ref{cor:semi_det_E} and matches exactly the region described by Sato's converse bound \cite[Prob. 8.8]{NIT}. 

Finally, the third class that we consider comprises $K$ user discrete memoryless broadcast channels. Our main result for this setting (Theorem \ref{thm:extension}) shows that  the finite-blocklength optimal probability of successful decoding under  full $(K+1)$-partite NS assistance is the same as if  each User $k$, $k\in[K]$ was additionally provided the desired messages of Users $k+1,\cdots,K$ as side-information. The side-information structure intuitively corresponds to a causal encoding order where the users' messages are encoded sequentially starting from User $K$ and ending with User $1$, so that in terms of the encoding for User $k$, the messages of previously encoded users with indices $k+1,\cdots, K$ comprise non-causal CSIT available to the encoder. Note that the ordering of users can be chosen arbitrarily.  
\section{Preliminaries}
\subsection{Notation}
$\mathbb{R}_+$ is the set of non-negative reals. $\mathbb{N}$ is the set of positive integers. For $n \in \mathbb{N}$, $[n] \triangleq \{1,2,\cdots, n\}$. 
$A_i^j$ denotes $[A_i,A_{i+1},\cdots, A_j]$, and  $A^n$ denotes $[A_1,A_2,\cdots, A_n]$. 
$\mathbb{F}_q$ is the finite field with order $q$, where $q$ is a power of a prime.  We write $g(n) = o(f(n))$ if $\lim_{n\to \infty} \frac{g(n)}{f(n)}$ $  = 0$.  The Cartesian product of sets $\mathcal{A}$ and $\mathcal{B}$ is $\mathcal{A} \times \mathcal{B}$.
The indicator function $\mathbb{I}(p)$ returns $1$ if the predicate $p$ is true, and $0$ otherwise. $\Pr(E)$ denotes the probability of the event $E$.  ``If and only if" is written as ``iff".

\subsection{Non-signaling correlations}
For a discrete set $\mathcal{X}$ of finite cardinality $|\mathcal{X}|<\infty$, let $\mathcal{P}(\mathcal{X})$ denote the set of probability mass functions on $\mathcal{X}$.  Let $\mathcal{P}(\mathcal{Y} \mid \mathcal{X})$ denote the set of conditional probability distributions where the input  variable is defined on the set $\mathcal{X}$ and the output variable is defined on the set $\mathcal{Y}$. Given $\mP_{Y_1Y_2\mid X} \in \mathcal{P}(\mathcal{Y}_1\times \mathcal{Y}_2\mid \mathcal{X})$, the conditional marginal distribution of $Y_1$ is defined as $\mP_{Y_1\mid X} \in \mathcal{P}(\mathcal{Y}_1\mid \mathcal{X})$ such that $\mP_{Y_1\mid X}(y_1\mid x) =\sum_{y_2\in \mathcal{Y}_2} \mP_{Y_1Y_2\mid X}(y_1,y_2\mid x)$ for all $y_1\in \mathcal{Y}_1,x\in \mathcal{X}$. The conditional marginal distribution of $Y_2$ is defined similarly.

For $K\geq 2$, let $\mathcal{P}_K^{\NS}(\mathcal{B}_1,\cdots, \mathcal{B}_K\mid \mathcal{A}_1,\cdots, \mathcal{A}_K) \subseteq \mathcal{P}(\mathcal{B}_1\times \cdots \times \mathcal{B}_K\mid \mathcal{A}_1\times \cdots \times \mathcal{A}_K)$ be the set of $K$-partite \emph{non-signaling} (NS) correlations (boxes), defined such that  $\mP_{B_1B_2\cdots B_K\mid A_1A_2\cdots A_K}\in \mathcal{P}_K^{\NS}(\mathcal{B}_1,\cdots, \mathcal{B}_K\mid \mathcal{A}_1,\cdots, \mathcal{A}_K)$ iff for all\footnote{In fact \cite{barrett2005nonlocal, masanes2006general} have shown that it suffices to let $\mathcal{U}=\{k\}$ for $k\in [K]$.} non-trivial bipartitions  $[K]=\mathcal{U} \cup \mathcal{V}$,  say $\mathcal{U} = \{u_1,u_2,\cdots, $ $u_m\}$, $\mathcal{V} = \{v_1,v_2,\cdots,v_n\}$, the (conditional) marginal distribution of $(B_{u_1},\cdots, B_{u_m})$ satisfies,
\begin{align} \label{eq:def_NS_condition}
  &\mP_{B_{u_1}\cdots B_{u_m} \mid A_{u_1}\cdots A_{u_m}A_{v_1}\cdots A_{v_n}}(b_{u_1},\cdots,b_{u_m}\mid a_{u_1},\cdots, a_{u_m}, a_{v_1},\cdots, a_{v_n}) \notag \\
  &=\mP_{B_{u_1}\cdots B_{u_m} \mid A_{u_1}\cdots A_{u_m}A_{v_1}\cdots A_{v_n}}(b_{u_1},\cdots,b_{u_m}\mid a_{u_1},\cdots, a_{u_m}, a_{v_1}',\cdots, a_{v_n}')   \\
  &\triangleq \mP_{B_{u_1}\cdots B_{u_m} \mid A_{u_1}\cdots A_{u_m}}(b_{u_1},\cdots,b_{u_m}\mid a_{u_1},\cdots, a_{u_m})
\end{align}
for all $\{b_{u_i},a_{u_i}\}_{i=1}^m, \{a_{v_j},a'_{v_j}\}_{j=1}^n $ with values chosen from their respective alphabets. In words, the condition says that the  marginal distribution of the outputs of any subset of parties only depends on the inputs of those parties. Note that in the last step we have redefined the marginal distribution for the parties indexed by $(u_1,\cdots, u_m)$ as $\mP_{B_1B_2\cdots B_K\mid A_1A_2\cdots A_K}$, to reflect that it is only a function of their own inputs and outputs.

\section{Problem Formulation I:  Channel with State}

Let $\mathcal{X}$, $\mathcal{Y}$ and $\mathcal{S}$ denote the alphabets for the input, output, and state, respectively. A channel with state is described by $\mN= (\mN_{Y\mid XS},\mP_S)$. Here $\mN_{Y\mid XS} \in \mathcal{P}(\mathcal{Y}\mid \mathcal{X}\times \mathcal{S})$ describes the output distribution of the channel given any input and  state, and $\mP_S \in \mathcal{P}(\mathcal{S})$ is the probability distribution of the states. 
A message $W\in \mathcal{M}$ is uniformly generated and made available to the transmitter, and needs to be communicated to the receiver.

\subsection{Coding for Channel with State}\label{sec:p2pMC}

\begin{figure}[htbp]
\center
\begin{tikzpicture}
  \node (Z) [rectangle, draw, thick, minimum width = 8cm, minimum height = 2cm] at (0,0) {};
  \draw[-, thick] (0,1) -- (0,-1);
  \node [rectangle, minimum height = 1cm, fill=gray!30] at (0,0) {\small $\mathsf{Z}(x,\hat{w}\mid [w,s],y)$};
  
  \node (IT) [rectangle, minimum width = 0.6cm,  fill=black]  at (-3,0.9) {};
  \node [below =0cm of IT] {\small $[w,s]$};
  
  \node (IR) [rectangle, minimum width=0.6cm, fill=black] at (3,-0.9) {};
  \node [above  =0cm of IR] {\small $y$};
  
  \node (OT)  [] at ($(IT.south)+(0,-1.6)$){\small $x$};
  
  \node (OR)  [] at ($(IR.north)+(0,1.5)$){\small $\hat{w}$};

  \node (N) [rectangle, draw, thick, minimum width = 3.2cm, minimum height = 1.5cm] at (0,-3) {};
  
  \node [rectangle, minimum height = 1cm, minimum width=2cm, fill=gray!30] at (0,-3) {\small $\mathsf{N}_{Y\mid XS} $};
  
  \node (X)[rectangle, fill=black, minimum height = 0.4cm] at ($(N.west)+(0.1,0.35)$){};
  
  \node (SN)[rectangle, fill=black, minimum height = 0.4cm] at ($(N.west)+(0.1,-0.35)$){};
  
  \node (S) at (-5,-1.8) [rectangle, draw, minimum height=1cm, thick] {$S\sim \mP_S$};
  
  \node (Y)[] at ($(N.east)+(-0.15,0)$){};
  
  \node (Tx) [above =1cm of IT] {\sc Transmitter};
  
  \node (Rx) [above =1cm of OR] {\sc Receiver};
  
  \node (C) [below =0cm of N] {\small \sc Channel with state $S$};
  
  \node  [above =0cm of Z] {\small \sc Non-Signaling Box};
  
  \draw[-{Latex[length=2.5mm]}, thick] (Tx) -- (IT) node[pos=0.5,left]{$[W,S]$};
  
  \draw[-{Latex[length=2.5mm]}, thick] (OT) |- (X) node[pos=0.25,left]{$X$};

  \draw[-{Latex[length=2.5mm]}, thick] (S.south) |- (SN);
  
  \draw[-{Latex[length=2.5mm]}, thick] (S.north) |- (Tx);
  
  \draw[-{Latex[length=2.5mm]}, thick] (Y) -| (IR) node[pos=0.75,right]{$Y$};

  \draw[-{Latex[length=2.5mm]}, thick] (OR) -- (Rx) node[pos=0.5,right]{$\hat{W}$};
  
\end{tikzpicture}
\caption{NS-assisted Coding scheme for channel with state.} \label{fig:CWS}
\end{figure}
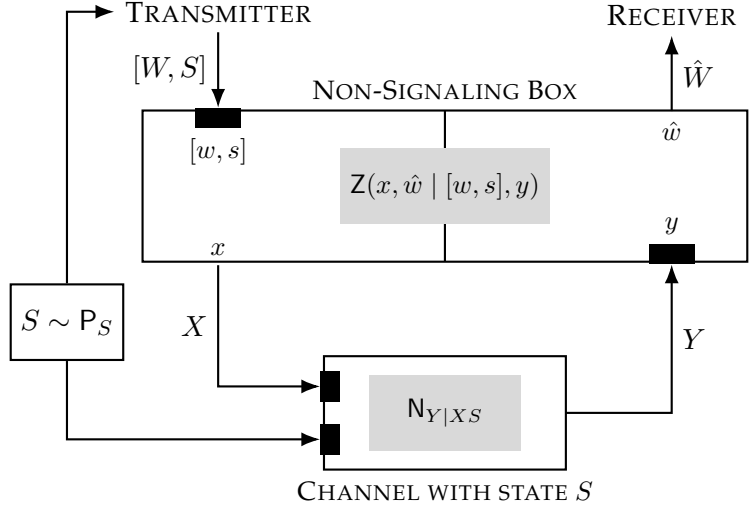

A coding scheme $\mZ$ for $\mN$ is specified by $\mZ\in \mathcal{P}(\mathcal{X} \times \mathcal{M} \mid \mathcal{M} \times \mathcal{S} \times \mathcal{Y})$. See Fig. \ref{fig:CWS}. Without loss of generality, assume $\mathcal{M} = [M]$, so that $M = |\mathcal{M}|$. We are  interested in two kinds of coding schemes (cf. \cite{matthews2012linear, cubitt2011zero}), denoted as $\mathcal{Z}^{\C}(M,\mN)$ and $\mathcal{Z}^{\NS}(M,\mN)$, for $M\in \mathbb{N}$, defined as follows.

\begin{itemize}[leftmargin=*]
  \item ``Classical" (C): $\mZ\in \mathcal{Z}^{\C}(M,\mN)$ iff there exists an encoder $\mE \in \mathcal{P}(\mathcal{X}\mid \mathcal{M}\times \mathcal{S} \times \mathcal{Q})$, a  decoder  $\mD\in \mathcal{P}(\mathcal{M}\mid \mathcal{Y}\times \mathcal{Q})$ and a distribution $\mP_Q\in \mathcal{P}(\mathcal{Q})$ such that $\mZ(x,\hat{w} \mid w,s,y) = \sum_{q\in \mathcal{Q}} \mP_Q(q) \mE(x\mid w,s,q) \mD(\hat{w}\mid y,q)$ for all $w\in \mathcal{M}, \hat{w} \in \mathcal{M}, x\in \mathcal{X}, y\in \mathcal{Y}, s\in \mathcal{S}$.
  \item ``Non-signaling" (NS): $\mZ \in \mathcal{Z}^{\NS}(M,\mN)$ iff $\mZ\in \mathcal{P}^{\NS}_2(\mathcal{X},\mathcal{M}\mid \mathcal{M}\times \mathcal{S}, \mathcal{Y})$.\footnote{According to the definition of $\mathcal{P}^{\NS}_2$, what this means explicitly is that $\sum_{x\in \mathcal{X}}\mZ(x,\hat{w}\mid w,s,y) = \sum_{x\in \mathcal{X}}\mZ(x,$ $\hat{w}\mid w',s',y)$ for all $\hat{w} \in \mathcal{M}, y\in \mathcal{Y}, w\in \mathcal{M}, w'\in \mathcal{M}, s\in \mathcal{S}, s'\in \mathcal{S}$, and $\sum_{\hat{w} \in \mathcal{M}}\mZ(x,\hat{w}\mid w,s,y) = \sum_{\hat{w} \in \mathcal{X}}\mZ(x,\hat{w}\mid w,s,y')$ for all $x\in \mathcal{X}, w\in \mathcal{M}, s\in \mathcal{S}, y\in \mathcal{Y}, y'\in \mathcal{Y}$.}
\end{itemize}
Since any classical coding scheme is non-signaling, $\mathcal{Z}^{\C}(M,\mN) \subseteq  \mathcal{Z}^{\NS}(M,\mN)$. 
A coding scheme $\mZ \in \mathcal{Z}^{\NS}(M,\mN)$ works as follows. The transmitter inputs $(W,S)$ to $\mZ$, and obtains the output $X$. This $X$ is sent through the channel $\mN$.  The user inputs $Y$ to $\mZ$, and obtains the decoded message $\hat{W}$.
For any given $\mN=(\mN_{Y\mid XS}, \mP_S)$ and $\mZ\in \mathcal{Z}^{\NS}(M,\mN)$, the joint distribution of all variables of interest is,
\begin{align}
  &\Pr(W=w,\hat{W}=\hat{w},S=s,X=x,Y=y)\notag\\
  &=\frac{1}{M} \mP_S(s)\Pr(\hat{W}=\hat{w},X=x,Y=y\mid W=w,S=s)\\
  &=\frac{1}{M} \mP_S(s) \mN_{Y\mid XS}(y\mid x,s) \mZ(x,\hat{w} \mid w, s, y)
\end{align}
where the last step holds because $\mZ\in \mathcal{Z}^{\NS}(M,\mN)$ (cf. \cite[Eq. (9)]{matthews2012linear}).

Define the probability of success $\eta(\mZ)$ associated with $\mZ\in \mathcal{Z}^{\NS}(M,\mN)$ as
\begin{align}
  &\eta(\mZ) \triangleq \Pr(W=\hat{W})\\
  &= \frac{1}{M}\sum_{w\in \mathcal{M},s\in \mathcal{S},  x\in \mathcal{X}, y\in \mathcal{Y}} \mP_S(s) \mN_{Y\mid XS}(y\mid x,s) \mZ(x,w\mid w, s, y) \label{eq:def_eta_cws}
\end{align}
Note that $\eta$ also depends on $(M, \mN_{Y\mid XS}, \mP_S)$, but we suppress these parameters for compact notation, as they can be inferred from the context.

\begin{remark}
  The coding schemes naturally allow the availability of the channel state information at the transmitter, commonly referred to as the setting of coding with CSIT. Any available channel state information at the receiver (CSIR) can be modeled by including it explicitly into the output of the channel. Therefore, the framework allows modeling of general channel state information (CSI) settings at both the transmitter and the receiver.
\end{remark}

\begin{remark} \label{rem:shared_randomness}
  Our formulation of classical coding schemes corresponds to coding with shared randomness (SR) in \cite{cubitt2011zero}.  In terms of maximal probability of success (equivalently, minimal probability of error), it suffices to consider a smaller set of coding schemes under SR, referred to as ``no-correlation" (NC) in \cite{cubitt2011zero, matthews2012linear}, for which $|\mathcal{Q}|=1$. It is not difficult to see that shared randomness cannot improve the probability of success for classical coding schemes, as a scheme with shared randomness is equivalently a convex combination of no-correlation (deterministic) schemes, and the probability of success of the scheme with shared randomness is equal to the (same) convex combination of the probabilities of success of those no-correlation schemes. Therefore, the maximal probability of success achievable by schemes with shared randomness is also  achievable with a no-correlation scheme.
\end{remark}

\subsection{Classical and Non-Signaling Assisted Capacity of Channel with State}
For $n\in \mathbb{N}$, let $\mN^{\otimes n}= (\mN_{Y^n\mid X^nS^n}, \mP_{S^n})$ represent $n$ uses of the channel $\mN=(\mN_{Y\mid XS}, \mP_S)$. Then $\mN^{\otimes n}$ is also a channel with state, defined such that,
\begin{align}
\begin{cases}
  \mN_{Y^n\mid X^nS^n}(y^n \mid x^n, s^n) = \prod_{i=1}^n \mN_{Y\mid XS}(y_i\mid x_i,s_i), \\
  \mP_{S^n}(s^n) = \prod_{i=1}^n \mP_S(s_i).
\end{cases}
\end{align}

A rate $R\in \mathbb{R}_+$ (measured in bits/channel-use) is said to be achievable classically (or with NS assistance) iff there exists a sequence of coding schemes $\mZ_n\in \mathcal{Z}^{\C}(M_n,\mN^{\otimes n})$ (or $\mZ_n \in \mathcal{Z}^{\NS}(M_n,\mN^{\otimes n})$) such that the limit of the probability of success $\lim_{n\to \infty} \eta(\mZ_n) = 1$ and the limit of the ratio $\lim_{n\to \infty}\frac{\log_2 M_n}{n} \geq R$. The classical capacity $C^{\C}(\mN)$ (or NS assisted capacity $C^{\NS}(\mN)$) is defined as the supremum of the set of rates that are achievable classically (or with NS assistance).

\begin{remark}
  A coding scheme with $n>1$, whether classical or with NS assistance, assumes non-causal CSIT, i.e., the output distribution for each $X_i$, $i\in[n]$ is allowed to depend on the entire state sequence $S^n$.
\end{remark}

\subsection{Results: NS Assisted Success Probability and Capacity for Channel with State}
Given a channel with state, $\mN=(\mN_{Y\mid XS}, \mP_S)$, let $\bar{\mN} = (\bar{\mN}_{YS_{\tR} \mid XS}, \mP_S)$ be another channel with state where the state is made available to the receiver (as an additional output, $S_{\tR}$). Formally, we define,
\begin{align}
  \bar{\mN}_{YS_{\tR}\mid XS}(y,s_{\tR} \mid x,s) \triangleq \mN_{Y\mid XS}(y\mid x,s) \times \mathbb{I}(s_{\tR}=s),\label{eq:defnbar}
\end{align}
for all $s_{\tR}\in \mathcal{S}, s \in \mathcal{S}, x\in \mathcal{X}, y\in \mathcal{Y}$. Note that $\bar{\mN}$ is still within the framework of channel with state. With this we are ready to present our first theorem.

\begin{theorem}[Virtual Signaling of CSIT] \label{thm:csit_tp}
  For every $M\in \mathbb{N}$,
  \begin{align}
    \max_{\mZ\in \mathcal{Z}^{\NS}(M,~\mN)}\eta(\mZ) = \max_{\mZ \in \mathcal{Z}^{\NS}(M,~\bar{\mN})}\eta(\mZ).
  \end{align}
\end{theorem}
\noindent Note that the LHS is the maximal (optimal) probability of success for NS assisted coding schemes over a channel with state, $\mN$, whereas the RHS is  for NS assisted coding schemes over $\bar{\mN}$. The theorem says that under \emph{NS assistance}, making the CSIT also available to the receiver cannot improve the maximal probability of success, i.e., any advantage in terms of optimal probability of success, of having CSIT also available to the receiver, is already available due to the NS assistance. In \emph{effect}, it is as if all CSIT was already \emph{signaled} to the receiver via NS assistance. This insight is what we refer to as \emph{virtual  signaling of CSIT}. Note that there is no \emph{actual} signaling of CSIT, i.e., the receiver may not explicitly acquire any knowledge of the CSIT. The \emph{`virtual' signaling of CSIT} refers only to the impact/effect of NS assistance on optimal success probability. The details of the proof of Theorem \ref{thm:csit_tp} are left to Appendix \ref{proof:csit_tp}. Let us provide a sketch here. \\[0.1cm]

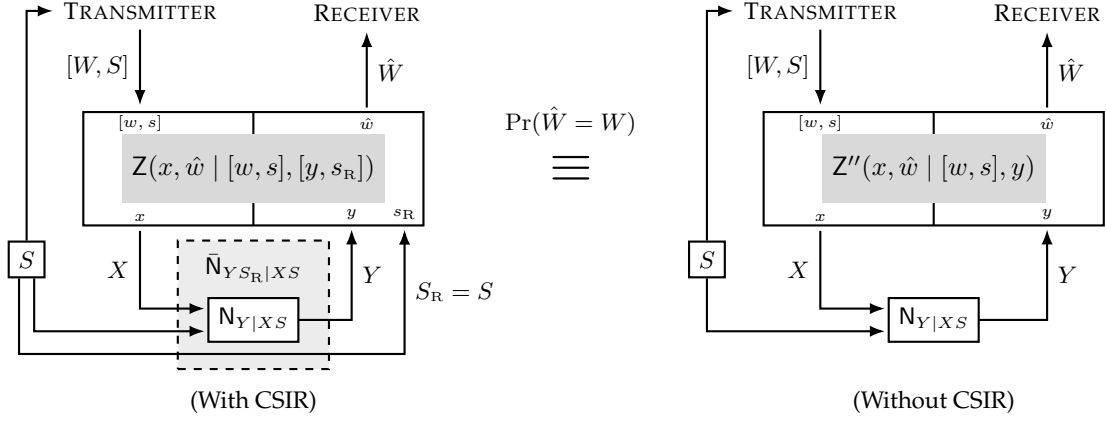
\begin{figure}
\center
\begin{tikzpicture}
\begin{scope}[shift={(0,0)}]
\node (Z) [rectangle, draw, thick, minimum width = 4.5cm, minimum height = 1.5cm] at (0,0) {};
  \draw[-, thick] (0,0.75) -- (0,-0.75);
  \node [rectangle, minimum height = 0.9cm, fill=gray!30] at (0,0) {\small $\mathsf{Z}(x,\hat{w}\mid [w,s],[y,s_{\tR}])$};
  
  \node (IT) [rectangle, minimum width = 0.6cm ]  at (-1.5,0.6) {\tiny $[w,s]$};

  \node (IR1) [rectangle, minimum width=0.6cm ] at (1.3,-0.6) {\tiny $y$};
  \node (IR2) [rectangle, minimum width=0.6cm ] at (2,-0.6) {\tiny $s_{\mathrm{R}}$};
  \node [above =0cm of IR] { };
  
  \node (OT)  [] at ($(IT.south)+(0,-1)$){\tiny $x$};
  
  \node (OR)  [] at ($(IR1.north)+(0.2,1.0)$){\tiny $\hat{w}$};

  \node (Nbar) [rectangle, draw, dashed, thick, minimum width = 2.0cm, minimum height = 1.7cm , fill=gray!15] at (0,-1.8) { };
  \node [above = -0.7cm of Nbar ] {\footnotesize $\bar{\mathsf{N}}_{YS_{\tR} \mid XS} $};
  
  \node (N) [rectangle, draw, thick, fill=white ] at (0,-2) {\footnotesize $\mathsf{N}_{Y \mid XS} $};
  
  \node (X)[rectangle, minimum height = 0.6cm] at ($(N.west)+(0.1,0.15)$){};
  
  \node (S) [rectangle, draw, thick ] at (-3,-1.2) {\footnotesize $S$};
  
  \node (SN)[rectangle, minimum height = 0.6cm] at ($(N.west)+(0.1,-0.15)$){};
  
  \node (Y)[] at ($(N.east)+(-0.15,0)$){};

  \node (Tx) [above =1cm of IT] {\footnotesize \sc Transmitter};
  
  \node (Rx) [above =1cm of OR] {\footnotesize \sc Receiver};
  
  \node (C) [below =0.5cm of N] {\footnotesize (With CSIR)};
  
  \draw[-{Latex[length=2mm]}, thick] (Tx) -- (IT) node[pos=0.5,left]{\footnotesize $[W,S]$};
  
  \draw[-{Latex[length=2mm]}, thick] (OT) |- (X) node[pos=0.25,left]{\footnotesize $X$};
 
  \draw[-{Latex[length=2mm]}, thick] (S.north) |- (Tx);

  \draw[-{Latex[length=2mm]}, thick] (Y) -| (IR1) node[pos=0.75,right]{\footnotesize $Y$};
  \draw[-{Latex[length=2mm]}, thick] ($(S.south)+(-0.1,0)$) -- ($(S.south)+(-0.1,-1)$) -| (IR2) node[pos=0.75,right]{\footnotesize $S_{\tR}=S$};
  
  \draw[-{Latex[length=2mm]}, thick] ($(S.south)+(+0.1,0)$) |- (SN)   node[pos=0.75,right]{ };

  \draw[-{Latex[length=2mm]}, thick] (OR) -- (Rx) node[pos=0.5,right]{\footnotesize $\hat{W}$};
\end{scope}

\begin{scope}[shift={(9,0)}]
\node (Z) [rectangle, draw, thick, minimum width = 4.5cm, minimum height = 1.5cm] at (0,0) {};
  \draw[-, thick] (0,0.75) -- (0,-0.75);
  \node [rectangle, minimum height = 0.9cm, fill=gray!30] at (0,0) {\small $\mathsf{Z}''(x,\hat{w}\mid [w,s],y)$};
  
  \node (IT) [rectangle, minimum width = 0.6cm ]  at (-1.5,0.6) {\tiny $[w,s]$};

  \node (IR) [rectangle, minimum width=0.6cm ] at (1.5,-0.6) {\tiny $y$};
 
  \node [above =0cm of IR] { };
  
  \node (OT)  [] at ($(IT.south)+(0,-1)$){\tiny $x$};
  
  \node (OR)  [] at ($(IR.north)+(0,1.0)$){\tiny $\hat{w}$};
  
  \node (N) [rectangle, draw, thick ] at (0,-2) {\footnotesize $\mathsf{N}_{Y \mid XS}$};
  
  \node (S) [rectangle, draw, thick ] at (-3,-1.2) {\footnotesize $S$};

  \node (X)[rectangle, minimum height = 0.6cm] at ($(N.west)+(0.1,0.15)$){};
  
  \node (SN)[rectangle, minimum height = 0.6cm] at ($(N.west)+(0.1,-0.15)$){};
  
  \node (Y)[] at ($(N.east)+(-0.15,0)$){};

  \node (Tx) [above =1cm of IT] {\footnotesize \sc Transmitter};
  
  \node (Rx) [above =1cm of OR] {\footnotesize \sc Receiver};
  
  \node (C) [below =0.5cm of N] {\footnotesize (Without CSIR)};
  
  \draw[-{Latex[length=2mm]}, thick] (Tx) -- (IT) node[pos=0.5,left]{\footnotesize $[W,S]$};
  
  \draw[-{Latex[length=2mm]}, thick] (OT) |- (X) node[pos=0.25,left]{\footnotesize $X$};
  
  \draw[-{Latex[length=2mm]}, thick] (S.north) |- (Tx);
  
  \draw[-{Latex[length=2mm]}, thick] (S.south) |- (SN);

  \draw[-{Latex[length=2mm]}, thick] (Y) -| (IR) node[pos=0.75,right]{\footnotesize $Y$};

  \draw[-{Latex[length=2mm]}, thick] (OR) -- (Rx) node[pos=0.5,right]{\footnotesize $\hat{W}$};
\end{scope}

\node (Equiv) at (4.2,0){\huge $\equiv$};
\node [above =0cm of Equiv]{\footnotesize $\Pr(\hat{W}=W)$};
\end{tikzpicture}
\caption{Equivalence between $\mZ$ and $\mZ''$ in terms of the probability of successful decoding.}
\label{fig:csit_teleportation}
\end{figure}

\noindent ({\it Proof Sketch}): The direction $\max_{\mZ \in \mathcal{Z}^{\NS}(M,\mN)}\eta(\mZ) \leq \max_{\mZ \in \mathcal{Z}^{\NS}(M,\bar{\mN})}\eta(\mZ)$ is obvious, as giving the CSIT to the receiver cannot hurt the optimal probability of success. To show the other direction, consider any NS assisted coding scheme $\mZ$ with CSIR, i.e.,
  \begin{align*}
    \mZ(x,\hat{w}\mid [w,s],[y,s_{\tR}]) \in \mathcal{Z}^{\NS}(M, \bar{\mN}).
  \end{align*}
  The square brackets on $[w,s]$ and $[y,s_{\tR}]$ simply group inputs to the NS box that correspond to the same party.
  Such a scheme needs both $Y$ and $S_{\tR}$ at the receiver.
  What we will do next is to construct in two steps an NS assisted coding scheme $\mZ''(x,\hat{w} \mid [w,s], y) \in \mathcal{Z}^{\NS}(M, \mN)$ that achieves the same probability of success as $\mZ$, but only needs $Y$ at the receiver. This will prove the other direction.
  First, let $\mathbb{C}_M$ be the cyclic permutation group operating on $[M]$. Then construct
  \begin{align}
    \mZ'(x,\hat{w}\mid [w,s], [y,s_{\tR}]) = \frac{1}{M}\sum_{\pi \in \mathbb{C}_M}\mZ(x, \pi(\hat{w})\mid [\pi(w),s], [y,s_{\tR}])\label{eq:twirl}
  \end{align}
  for all $(w,\hat{w}, x,y,s, s_{\tR})\in \mathcal{M}^2\times \mathcal{X}\times \mathcal{Y} \times \mathcal{S}^2$. 
  We argue that $\eta(\mZ)=\eta(\mZ')$, which can be seen from the fact that $\mZ'$ is simply a convex combination (with uniform coefficients $1/M$) of $M$ schemes, each of which is $\mZ$ after relabeling the message indices according to the cyclic permutation $\pi$, and that relabeling does not affect the probability of success since the message is uniformly generated. 
  
  Next, construct
  \begin{align}
    \mZ''(x,\hat{w}\mid [w,s],y) = \mZ'(x,\hat{w} \mid [w,s],[y,s])
  \end{align}
  for all $(w, \hat{w}, x,y,s)\in \mathcal{M}^2\times \mathcal{X}\times \mathcal{Y} \times \mathcal{S}$. 
  One should verify that $\eta(\mZ'')=\eta(\mZ')$ simply by definition. The last thing is to verify that $\mZ''$ is indeed non-signaling, in particular that $\sum_{x}\mZ''(x,\hat{w} \mid [w,s],y) =1/M$ and that $\sum_{\hat{w}}\mZ''(x,\hat{w} \mid [w,s],y)$ only depends on $(x,w,s)$ (in fact only $(x,s)$). We illustrate the idea in Fig. \ref{fig:csit_teleportation}. The details appear in Appendix \ref{proof:csit_tp}.
  \hfill\qed

\begin{remark}
  The construction of $\mZ'$ in \eqref{eq:twirl} is similar to the twirling argument of \cite{cubitt2011zero, matthews2012linear}, but here it suffices to consider only the cyclic permutations instead of all permutations.
\end{remark}

\begin{theorem}[NS assisted capacity of channel with state] \label{thm:capacity_cws}
  The NS assisted capacity of $\mN=(\mN_{Y\mid X}, \mP_S)$ is,
  \begin{align} \label{eq:capacity_NS_cws}
    C^{\NS}(\mN) &= \max_{\mP_{X\mid S}} I(X;Y\mid S)\\
    &=C^{\C}(\bar\mN)\\
    &=C^{\NS}(\bar\mN), 
  \end{align}
  where the maximum is taken over all  $\mP_{X\mid S}\in \mathcal{P}(\mathcal{X} \mid \mathcal{S})$, and $\bar{\mN} = (\bar{\mN}_{YS_{\tR} \mid XS}, \mP_S)$ is  defined in \eqref{eq:defnbar}.
\end{theorem}
\noindent The proof is provided in Appendix \ref{proof:capacity_cws}. Recall that in the channel $\bar{\mN}$ the state $S$ is available not only to the transmitter -- as CSIT, but also to the receiver -- as CSIR. It is known that the RHS of \eqref{eq:capacity_NS_cws} represents the classical capacity of $\bar{\mN}$, i.e., $C^{\C}(\bar\mN)=\max_{\mP_{X|S}}I(X;Y|S)$ \cite[Sec. 7.4.1]{NIT}. From Theorem \ref{thm:csit_tp} it follows immediately that $C^{\NS}(\mN)=C^{\NS}(\bar\mN)$.  The achievability argument for \eqref{eq:capacity_NS_cws} also follows immediately, as $C^{\NS}(\mN) = C^{\NS}(\bar{\mN}) \geq C^{\C}(\bar{\mN})= \mbox{RHS of } \eqref{eq:capacity_NS_cws}$. To complete the proof of Theorem \ref{thm:capacity_cws}, we only need to show the converse, that the RHS of \eqref{eq:capacity_NS_cws} serves also as an upper bound for $C^{\NS}(\mN)$.

This result stands in contrast with the capacity for the corresponding classical setting, known as the Gelfand-Pinsker Theorem \cite[Thm. 7.3]{NIT}, \cite{gel1980coding},
\begin{align}
  C^{\C}(\mN) = \max_{\mP_{X U\mid S}} \big(I(U;Y) - I(U;S)\big),
\end{align}
where the maximum is taken over all joint distributions of $(X,U)$ conditioned on $S$, $\mP_{XU\mid S} \in \mathcal{P}(\mathcal{X}\times \mathcal{U}\mid \mathcal{S})$ with $|\mathcal{U}|\leq \min\{|\mathcal{X}||\mathcal{S}|, |\mathcal{Y}|+ |\mathcal{S}|-1\}$. The classical capacity of $\mN$ can be strictly smaller than that of $\bar{\mN}$. In fact, in general the gap between them can be arbitrarily large (see Example \ref{ex:fadedirt}). However, Theorems \ref{thm:csit_tp} and \ref{thm:capacity_cws} imply that $\mN$ and $\bar{\mN}$ have  the same NS assisted capacity. Moreover, as an even stronger implication, $\mN$ and $\bar{\mN}$ have the same NS assisted optimal probability of success over any arbitrary fixed finite number of channel-uses.

\begin{example}[Fading dirt\cite{Yao_Jafar_NS_DoF}] \label{ex:fadedirt}
  The $\mathbb{F}_q$ fading dirt channel (analogous to the wireless fading dirt channel \cite{Rini_Shamai_fading_dirt}) $\mN=(\mN_{Y\mid XS}, \mP_S)$ is defined such that $Y= (\bar{Y},G)$, $\bar{Y}= X+G\times S$, where $X,\bar{Y}\in \mathbb{F}_q$, $G$ and $S$ are each uniformly distributed over $\mathbb{F}_q$, and $G$ is independent of $(X,S)$. All operations are over $\mathbb{F}_q$. Since $\log_2 q \geq \max_{\mP_{X\mid S}}H(X)  \geq \max_{\mP_{X\mid S}}I(X;Y\mid S) = \max_{\mP_{X\mid S}}I(X;\bar{Y},G \mid S)   \geq  \max_{\mP_{X\mid S}}H(X\mid S)= \log_2 q$, Theorem \ref{thm:capacity_cws} implies that $C^{\NS}(\mN) = \log_2 q$, recovering a result of \cite{Yao_Jafar_NS_DoF}. In fact, as shown in \cite{Yao_Jafar_NS_DoF}, in this case, the gap between $C^{\NS}(\mN) = C^{\C}(\bar\mN)$ and $C^{\C}(\mN)$ is unbounded as $q\rightarrow\infty$.
\end{example}

\section{Problem Formulation II: NS assistance for $2$-User Broadcast} \label{sec:BC}
In this section, we consider NS assisted coding for 2-user broadcast channels (BC), i.e., a broadcast channel with $2$ receivers.  The main result in this section is the complete characterization of the NS assisted capacity (region) of all 2-user BCs with bipartite NS assistance established between the transmitter and one of the receivers. Capacity is also characterized for a particular set of 2-user BCs with NS assistance among all parties. We begin with the formal definition of the framework of coding schemes.

A (2-user) broadcast channel is described by $\mN  \in \mathcal{P}(\mathcal{Y}_1\times \mathcal{Y}_2\mid \mathcal{X})$, which specifies the output distribution (at the $2$ users/receivers) given the channel's input. Two independent messages $W_1,W_2$ are generated at the transmitter, such that $W_i$ is to be communicated to User $i$, for $i\in \{1,2\}$.

\subsection{Coding schemes} \label{subsec:coding_schemes_NSBC}
\begin{figure}[htbp]
\center
\begin{tikzpicture}
  \node (Z) [rectangle, draw, thick, minimum width = 8cm, minimum height = 2cm] at (0,0) {};
  \draw[-, thick] (0,1) -- (0,-1);
  \draw[-, thick] (2,1) -- (2,-1);
  \node [rectangle, minimum height = 0.7cm, fill=gray!30] at (1,0) {\small $\mZ(x,\hat{w}_1,\hat{w}_2 \mid [w_1,w_2],y_1,y_2)$};
  
  \node (IT) [rectangle, minimum width = 0.6cm,  fill=black]  at (-3,0.9) {};
  \node [below =0cm of IT] {\small $[w_1,w_2]$};
  
  \node (IR1) [rectangle, minimum width=0.6cm, fill=black] at (1,-0.9) {};
  \node [above  =-0.1cm of IR1] {\small $y_1$};
  
  \node (IR2) [rectangle, minimum width=0.6cm, fill=black] at (3,-0.9) {};
  \node [above  =-0.1cm of IR2] {\small $y_2$};
  
  \node (OT)  [] at ($(IT.south)+(0,-1.6)$){\small $x$};
  
  \node (OR1)  [] at ($(IR1.north)+(0,1.5)$){\small $\hat{w}_1$};
  \node (OR2)  [] at ($(IR2.north)+(0,1.5)$){\small $\hat{w}_2$};

  \node (N) [rectangle, draw, thick, minimum width = 2.5cm, minimum height = 1.5cm] at (-1,-3) {};
  
  \node [rectangle, minimum height = 1cm, fill=gray!30] at (-1,-3) {\small $\mathsf{N}_{Y_1Y_2\mid X}$};
  
  \node (X)[rectangle, fill=black, minimum height = 0.6cm] at ($(N.west)+(0.1,0)$){};
  
  \node (Y1)[] at ($(N.east)+(-0.15,+0.5)$){};
  \node (Y2)[] at ($(N.east)+(-0.15,-0.5)$){};

  \node (Tx) [above =1cm of IT] {\sc Transmitter};
  
  \node (Rx1) [above =1cm of OR1] {\sc User $1$};
  \node (Rx2) [above =1cm of OR2] {\sc User $2$};
  
  \node (C) [below =0cm of N] {\small \sc Broadcast channel};
  
  \draw[-{Latex[length=2.5mm]}, thick] (Tx) -- (IT) node[pos=0.5,left]{$[W_1,W_2]$};
  
  \draw[-{Latex[length=2.5mm]}, thick] (OT) |- (X) node[pos=0.25,left]{$X$};
  
  \draw[-{Latex[length=2.5mm]}, thick] (Y1) -| (IR1) node[pos=0.66,right]{$Y_1$};
  \draw[-{Latex[length=2.5mm]}, thick] (Y2) -| (IR2) node[pos=0.795,right]{$Y_2$};
  
  \draw[-{Latex[length=2.5mm]}, thick] (OR1) -- (Rx1) node[pos=0.5,right]{$\hat{W}_1$};
  \draw[-{Latex[length=2.5mm]}, thick] (OR2) -- (Rx2) node[pos=0.5,right]{$\hat{W}_2$};
  
\end{tikzpicture}
\caption{NS-assisted coding scheme for broadcast channel} \label{fig:BC}
\end{figure}
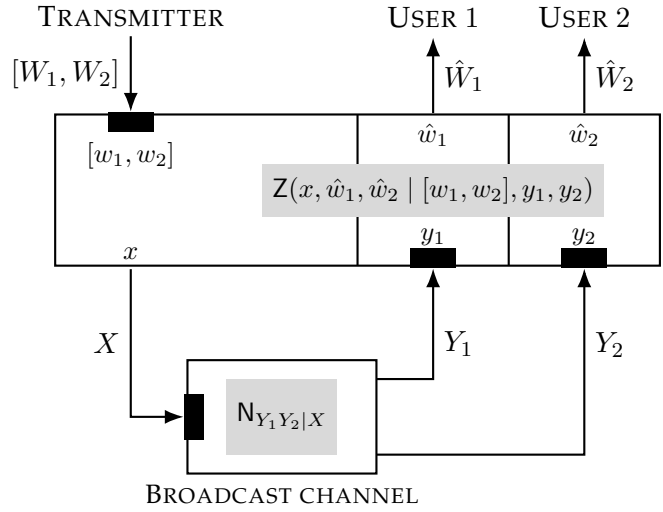

A coding scheme $\mZ$ over the BC is specified by $\mZ\in \mathcal{P}(\mathcal{X}\times \mathcal{M}_1\times \mathcal{M}_2 \mid \mathcal{M}_1\times \mathcal{M}_2\times \mathcal{Y}_1\times \mathcal{Y}_2)$. See Fig. \ref{fig:BC}. Without loss of generality, assume $\mathcal{M}_i = [M_i]$ for $i\in \{1,2\}$, where $M_i = |\mathcal{M}_i|$. We are interested in the following $5$ subsets of coding schemes, denoted as $\mathcal{Z}^{\C}(M_1,M_2,\mN), \mathcal{Z}^{\NS\mbox{-}i}(M_1,M_2,\mN)$ for $i\in \{0,1,2\}$ and $\mathcal{Z}^{\NS}(M_1,M_2,\mN)$, defined as follows. We will not write the domain of the variables explicitly as it is clear from the context.
\begin{itemize}[leftmargin=*]
  \item ``Classical" (C): $\mZ\in \mathcal{Z}^{\C}(M,\mN)$ iff there exists an encoder $\mE\in \mathcal{P}(\mathcal{X}\mid \mathcal{M}_1\times \mathcal{M}_2)$ and two separate decoders $\mD_i \in \mathcal{P}(\mathcal{M}_i \mid \mathcal{Y}_i)$ for $i\in \{1,2\}$, such that $\mZ(x,\hat{w}_1,\hat{w}_2 \mid [w_1,w_2],y_1,y_2)= \mE(x\mid w_1,w_2) \times \mD_1(\hat{w}_1\mid y_1) \times \mD_2(\hat{w}_2\mid y_2)$,  concisely written as $\mZ=\mE\times \mD_1\times \mD_2$.
  \item ``Bipartite NS assistance between the two users" (NS-0): $\mZ\in \mathcal{Z}^{\NSz}(M,\mN)$ iff there exists an encoder $\mE\in \mathcal{P}(\mathcal{X} \mid \mathcal{M}_1\times \mathcal{M}_2)$ and an  NS assisted  bipartite decoder $\mD\in \mathcal{P}^{\NS}_2(\mathcal{M}_1, \mathcal{M}_2\mid \mathcal{Y}_1, \mathcal{Y}_2)$, such that $\mZ(x,\hat{w}_1,\hat{w}_2\mid [w_1,w_2], y_1,y_2) = \mE(x\mid w_1,w_2)\times \mD(\hat{w}_1,\hat{w}_2\mid y_1,y_2)$, i.e., $\mZ=\mE\times \mD$.  This setting allows a bipartite NS box shared only between the two receivers. There is no communication between the two receivers.  
  \item ``Bipartite NS assistance between the transmitter and User $i$" (NS-$i$): For $i\in \{1,2\}$, \\$\mZ\in\mathcal{Z}^{\NS\text{-}{i}}(M,\mN)$ iff there exists $\mZ_i \in \mathcal{P}^{\NS}_2(\mathcal{X}, \mathcal{M}_i \mid \mathcal{M}_1\times \mathcal{M}_2, \mathcal{Y}_i)$ and a separate decoder $\mD_j\in \mathcal{P}(\mathcal{M}_j\mid \mathcal{Y}_j)$ ($\{j\}=\{1,2\}\setminus \{i\}$), such that $\mZ(x,\hat{w}_1,\hat{w}_2\mid [w_1,w_2],y_1,y_2) = \mZ_i(x,\hat{w}_i\mid [w_1,w_2], y_i)\times D_j(\hat{w}_j\mid y_j)$, i.e., $\mZ=\mZ_i\times \mD_j$.
  \item ``Full NS assistance between three parties" (NS): $\mZ\in \mathcal{Z}^{\NS}(M,\mN)$ iff $\mZ\in \mathcal{P}^{\NS}_3(\mathcal{X},\mathcal{M}_1,\mathcal{M}_2\mid \mathcal{M}_1\times\mathcal{M}_2,\mathcal{Y}_1,\mathcal{Y}_2)$. 
\end{itemize} 
Clearly, $\mathcal{Z}^{\C}(M_1,M_2,\mN) \subseteq \mathcal{Z}^{\NS\text{\tiny-}{i}}(M_1,M_2,\mN) \subseteq \mathcal{Z}^{\NS}(M_1,M_2,\mN)$ for $i\in \{0,1,2\}$. 
A coding scheme $\mZ \in \mathcal{Z}^{\NS}(M_1,M_2,\mN)$ works as follows. The transmitter inputs $(W_1,W_2)$ to $\mZ$, and obtains the output $X$. This $X$ is sent through the channel $\mN$. For $k\in \{1,2\}$, User $k$ inputs $Y_k$ to $\mZ$, and obtains the decoded message $\hat{W}_k$.
Given $\mN\in \mathcal{P}(\mathcal{Y}_1\times \mathcal{Y}_2\mid \mathcal{X})$ and $\mZ\in \mathcal{Z}^{\NS}(M_1,M_2,N)$, the joint  distribution of all variables of interest is,
\begin{align}
  &\Pr(\hat{W}_i=\hat{w}_i, W_j=w_j,X=x,Y_k=y_k,\forall (i,j,k)\in\{1,2\}^3) \\
  &=\frac{1}{M_1M_2} \Pr(\hat{W}_1=\hat{w}_1, \hat{W}_2= \hat{w}_2, X=x,Y_1=y_1,Y_2=y_2 \mid W_1=w_1,W_2=w_2)\\
  &=\frac{1}{M_1M_2}  \mZ(x,\hat{w}_1,\hat{w}_2\mid [w_1,w_2],y_1,y_2)\times \mN(y_1,y_2\mid x)
\end{align}
where the last step holds because $\mZ\in \mathcal{Z}^{\NS}(M_1,M_2,\mN)$ (cf. \cite[Sec. IV.B]{fawzi2024broadcast}).

The \emph{joint} probability of success associated with $\mZ\in \mathcal{Z}^{\NS}(M_1,M_2,\mN)$ is defined as
\begin{align}
  &\eta(\mZ) \triangleq \Pr(W_1=\hat{W}_1,W_2=\hat{W}_2) \notag \\
  &= \frac{1}{M_1M_2} \sum_{w_1,w_2,x,y_1,y_2} \mZ(x,w_1,w_2\mid [w_1,w_2],y_1,y_2)\times \mN(y_1,y_2\mid x)
\end{align}

The \emph{individual} probability of success associated with $\mZ\in \mathcal{Z}^{\NS}(M_1,M_2,\mN)$ for message $W_i$, $i\in \{1,2\}$ is similarly defined as, (with $\{j\}=\{1,2\}\setminus \{i\}$)
\begin{align}
  &\eta_i(\mZ) \triangleq \Pr(W_i=\hat{W}_i) \notag \\
  &= \frac{1}{M_1M_2} \sum_{w_i,w_j,\hat{w}_j,x,y_1,y_2}  N(y_1,y_2\mid x) \times 
  \begin{cases}
     \mZ(x,w_1,\hat{w}_2\mid [w_1,w_2],y_1,y_2), & i=1\\
     \mZ(x,\hat{w}_1,w_2\mid [w_1,w_2],y_1,y_2), & i=2
  \end{cases}
\end{align}
\begin{remark}
  Shared randomness is not considered for simplicity for the classes ''$\mathrm{C}$", ''$\mathrm{NS}\text{-}i$," $i\in \{0,1,2\}$. This is without loss of generality when considering the optimal probability of success, as is explained in Remark \ref{rem:shared_randomness}.
\end{remark}

\subsection{Achievable rate pairs and capacity region}
For $n\in \mathbb{N}$, let $\mN^{\otimes n} \in \mathcal{P}(\mathcal{Y}_1^n \times \mathcal{Y}_2^n \mid \mathcal{X}^n)$ be $n$ uses of the broadcast channel $\mN$, defined as
\begin{align}
  \mN^{\otimes n}(y_1^n, y_2^n \mid x^n) = \prod_{i=1}^n \mN(y_{1,i}, y_{2,i}\mid x_i).
\end{align}
For $\PH\in \{``\mathrm{C}", ``\mathrm{NS}\text{-}0", ``\mathrm{NS}\text{-}1", ``\mathrm{NS}\text{-}2", ``\mathrm{NS}"\}$,
a rate pair $(R_1,R_2) \in \mathbb{R}_+^2$ is said to be achievable
\begin{itemize}[leftmargin=*]
  \item classically, iff $\PH = ``\mathrm{C}"$
  \item by bipartite NS assistance between the two users iff $\PH=``\mathrm{NS}\text{-}0"$
  \item by bipartite NS assistance between the transmitter and User 1 iff $\PH=``\mathrm{NS}\text{-}1"$
  \item by bipartite NS assistance between the transmitter and User 2 iff $\PH=``\mathrm{NS}\text{-}2"$
  \item by full NS assistance among the three parties iff $\PH=``\mathrm{NS}"$
\end{itemize}
and iff there exists a sequence of coding schemes $\mZ_n\in \mathcal{Z}^{\PH}(M_{1,n},M_{2,n},\mN^{\otimes n})$ such that $\lim_{n\to \infty} \eta(\mZ) = 1$ and $\lim_{n\to \infty}\frac{\log_2(M_{i,n})}{n}\geq R_i$ for $i\in \{1,2\}$. Equivalently, one can replace the first condition with $\lim_{n\to \infty}\eta_i(\mZ) = 1$ for $i\in \{1,2\}$, as one  implies the other.

The capacity region $\mathcal{C}^{\PH}$ for $\PH\in  \{``\mathrm{C}", ``\mathrm{NS}\text{-}0", ``\mathrm{NS}\text{-}1", ``\mathrm{NS}\text{-}2", ``\mathrm{NS}" \}$ is defined as the closure of the set of  rate pairs achievable by the corresponding class of coding schemes.

\subsection{Coding with side information}
For our purpose, it is useful to introduce another setting where User $1$ is   provided with the message $W_2$, a special case of the problem studied in \cite{kramer2007capacity}. Let us refer to this setting as coding with side information (at User $1$). Given a broadcast channel $\mN$, a coding scheme with side information is specified by $\mZ\in \mathcal{P}(\mathcal{X}\times \mathcal{M}_1\times \mathcal{M}_2 \mid [\mathcal{M}_1\times \mathcal{M}_2] \times [\mathcal{Y}_1\times \mathcal{M}_2]\times \mathcal{Y}_2)$. The difference now is that User $1$ decodes the message $W_1$ from $Y_1$ and $W_2$ together, since $W_2$ is given as its side information. We can similarly introduce the classes of coding schemes according to the availability of NS assistance as in Section \ref{subsec:coding_schemes_NSBC}. In particular, 

\begin{itemize}[leftmargin=*]
  \item $\mathcal{Z}^{\C,\SIa}$ is defined such that $\mZ\in \mathcal{Z}^{\C,\SIa}$ iff there exists an encoder $\mE\in \mathcal{P}(\mathcal{X}\mid \mathcal{M}_1\times \mathcal{M}_2)$ and two separate decoders $\mD_1 \in \mathcal{P}(\mathcal{M}_1 \mid \mathcal{Y}_1\times \mathcal{M}_2)$, $\mD_2 \in \mathcal{P}(\mathcal{M}_2 \mid \mathcal{Y}_2)$ such that $\mZ(x,\hat{w}_1,\hat{w}_2 \mid [w_1,w_2],[y_1,w_2^{\tR}],y_2)= \mE(x\mid w_1,w_2) \times \mD_1(\hat{w}_1\mid y_1,w_2^{\tR}) \times \mD_2(\hat{w}_2\mid y_2)$, i.e., $\mZ=\mE\times \mD_1\times \mD_2$.
  \item $\mathcal{Z}^{\NSa,\SIa}$ is defined such that $\mZ\in \mathcal{Z}^{\NSa,\SIa}$ iff there exists $\mZ_1\in \mathcal{P}^{\NS}_2(\mathcal{X},\mathcal{M}_1\mid \mathcal{M}_1\times \mathcal{M}_2,\mathcal{Y}_1\times \mathcal{M}_2)$ and a separate decoder $\mD_2\in \mathcal{P}(\mathcal{M}_2\mid \mathcal{Y}_2)$, such that $\mZ(x,\hat{w}_1,\hat{w}_2\mid [w_1,w_2],[y_1,w_2^{\tR}],y_2) = \mZ_1(x,\hat{w}_1\mid [w_1,w_2],[y_1,w_2^{\tR}])\times \mD_2(\hat{w}_2\mid y_2)$, i.e., $\mZ=\mZ_1\times \mD_2$. 
  \item $\mathcal{Z}^{\NS,\SIa}=\mathcal{P}^{\NS}_3(\mathcal{X},\mathcal{M}_1,\mathcal{M}_2\mid \mathcal{M}_1\times \mathcal{M}_2,\mathcal{Y}_1\times \mathcal{M}_2,\mathcal{Y}_2)$.
\end{itemize}
The extra label ``$\mathrm{SI}\text{-}1$" indicates that the side information $W_2$ is available to User 1. One should note the extra input at User $1$ in the coding schemes. 
A coding scheme $\mZ \in \mathcal{Z}^{\NS,\SIa}(M_1,M_2,\mN)$ works as follows. The transmitter inputs $(W_1,W_2)$ to $\mZ$, and obtains the output $X$. This $X$ goes through the channel $\mN$.  User $1$ inputs $(Y_1,W_2)$ to $\mZ$, and obtains the decoded message $\hat{W}_1$. User $2$ inputs $Y_2$ to $\mZ$, and obtains the decoded message $\hat{W}_2$.
For $\mZ\in \mathcal{Z}^{\NS,\SIa}(M_1,M_2,\mN)$, 
the joint probability of success is
\begin{align}
  &\eta(\mZ) = \Pr(W_1=\hat{W}_1, W_2=\hat{W}_2)\notag \\
  &=\frac{1}{M_1M_2} \sum_{w_1,w_2,x,y_1,y_2}   \mZ(x,w_1,w_2\mid [w_1,w_2],[y_1,w_2],y_2) \times \mN(y_1,y_2\mid x)
\end{align}
and the individual probability of success is, for $i\in \{1,2\}$, $\{j\}=\{1,2\}\setminus\{i\}$,
\begin{align}
  &\eta_i(\mZ) = \Pr(W_i=\hat{W}_i) \notag \\
  &= \frac{1}{M_1M_2} \sum_{w_i,w_j,\hat{w}_j,x,y_1,y_2}  \mN(y_1,y_2\mid x) \times 
  \begin{cases}
     \mZ(x,w_1,\hat{w}_2\mid [w_1,w_2],[y_1,w_2],y_2), & i=1\\
     \mZ(x,\hat{w}_1,w_2\mid [w_1,w_2],[y_1,w_2],y_2), & i=2
  \end{cases}
\end{align}
We accordingly define the capacity regions with side information at User $1$ as $\mathcal{C}^{\PH,\SIa}$, for $\PH\in \{``\mathrm{C}", ``\mathrm{NS}\text{-}0", ``\mathrm{NS}\text{-}1", ``\mathrm{NS}\text{-}2", ``\mathrm{NS}"\}$.

\subsection{Results: NS Assisted $2$-User BC}
It suffices to present the results on the classes $``\mathrm{C}", ``\mathrm{NS}\text{-}0", ``\mathrm{NS}\text{-}1"$ and $``\mathrm{NS}"$ because the result of $``\mathrm{NS}\text{-}2"$ can be easily obtained by that of $``\mathrm{NS}\text{-}1"$ by exchanging the two users' indices. Given a broadcast channel $\mN$, it is useful to define the following two regions of rate pairs $(R_1,R_2)$. The first region is
\begin{align}
  \mathcal{R}_{\KS}(\mN) \triangleq \bigcup_{\mP_{XU}}\left\{
  \begin{array}{l}
    (R_1,R_2)\colon \\
    R_1 \leq I(X;Y_1) \\
    R_2 \leq I(U;Y_2) \\
    R_1 + R_2 \leq I(X;Y_1\mid U) + I(U;Y_2)
  \end{array}
  \right.
\end{align}
where the union is over all $\mP_{XU}\in \mathcal{P}(\mathcal{X}\times \mathcal{U})$. It suffices to consider $|\mathcal{U}| <  |\mathcal{X}|$.
We point out that $\mathcal{R}_{\KS}$ is the capacity region for  classical coding schemes when  $W_2$ is known to User $1$. This is established in \cite[Example 2]{kramer2007capacity}, and  `KS' stands for the initials of the authors (Kramer, Shamai).

The second region is,
\begin{align}
  \mathcal{R}_{\Sato}(\mN) \triangleq   \bigcup_{\mP_X} \left\{
  \begin{array}{l}
    (R_1,R_2)\colon \\
    R_1 \leq I(X;Y_1) \\
    R_2 \leq I(X;Y_2) \\
    R_1 + R_2 \leq \min_{\mN'} I(X;Y_1',Y_2')
  \end{array}
  \right.
\end{align}
where the union is over all $\mP_X\in \mathcal{P}(\mathcal{X})$, and the minimization of the bound on $R_1+R_2$ is taken over all channels $\mN'\in \mathcal{P}(\mathcal{Y}_1\times \mathcal{Y}_2\mid \mathcal{X})$ (with $Y_1',Y_2'$ denoting the outputs of $\mN'$) such that $\sum_{y_2}\mN'(y_1,y_2\mid x) = \sum_{y_2}\mN(y_1,y_2\mid x)$ and $\sum_{y_1}\mN'(y_1,y_2\mid x) = \sum_{y_1}\mN(y_1,y_2\mid x)$, meaning that $\mN$ and $\mN'$ have the same marginal distributions for each user individually.
We point out that $\mathcal{R}_{\Sato}$ serves as an outer bound of the classical capacity region $\mathcal{C}^{\C}(\mN)$ for any broadcast channel $\mN$ \cite[Prob. 8.8]{NIT}, but it is generally not tight.

Fawzi and Ferm{\'e} established in \cite[Thm. 1]{fawzi2024broadcast} that if NS assistance  is only available between the two receivers, then it does not improve the capacity region, i.e.,
\begin{align}
  \mathcal{C}^{\C}(\mN) = \mathcal{C}^{\NSz}(\mN).
\end{align}
This follows from the stronger result, also established in \cite{fawzi2024broadcast}, that the individual probabilities of success $(\eta_1(\mZ), \eta_2(\mZ))$ cannot be improved by the receiver-side NS assistance.\footnote{It is not difficult to see that if the probabilities of success could be improved  by NS-assistance, then it would allow signaling between the two receivers who otherwise have no channel between them, thus violating the NS constraint.} 

Since the classical capacity of a general broadcast channel $\mN$ is unknown, $\mathcal{C}^{\NSz}(\mN)$ is also open for general $\mN$. Fortunately, our next result  finds the solution for $\mathcal{C}^{\NSa}(\mN)$ for general $\mN$.

\begin{theorem}[NS-1] \label{thm:bipartite_NSBC}
  For a $2$ user broadcast channel $\mN$, with bipartite NS assistance between the transmitter and User $1$, the capacity region  is
  \begin{align}
    \mathcal{C}^{\NSa}(\mN) = \mathcal{C}^{\NSa,\SIa}(\mN) = \mathcal{C}^{\C,\SIa}(\mN) = \mathcal{R}_{\KS}(\mN).
  \end{align}
\end{theorem}
\noindent Since $\mathcal{R}_{\KS}(\mN)$ has a computable form, Theorem \ref{thm:bipartite_NSBC} provides a computable form for $\mathcal{C}^{\NSa}(\mN)$ and $\mathcal{C}^{\NSa,\SIa}(\mN)$ in general.  
The proof is presented in Appendix \ref{proof:bipartite_NSBC}. 
Recall that $C^{\C,\SIa}(\mN) = \mathcal{R}_{\KS}(\mN)$ is established in \cite{kramer2007capacity}, and $\mathcal{C}^{\NSa,\SIa}(\mN) \supseteq \mathcal{C}^{\C,\SIa}(\mN)$ by definition. Therefore, it suffices to prove $\mathcal{C}^{\NSa}(\mN)=\mathcal{C}^{\NSa,\SIa}(\mN)$, and that $\mathcal{C}^{\NSa}(\mN) \subseteq \mathcal{R}_{\KS}(\mN)$.

Next, we shift our focus to the fully NS assisted capacity region.

\begin{theorem}[Sato's outer bound] \label{thm:Sato}
  For a broadcast channel $\mN$, the capacity region with full NS assistance satisfies
  \begin{align}
    \mathcal{C}^{\NS}(\mN) \subseteq \mathcal{R}_{\Sato}(\mN).
  \end{align}
\end{theorem}
\noindent The proof is given in Appendix \ref{proof:Sato}. Note that Theorem 4 says that the well-known Sato's outer bound for the classical capacity region holds for fully NS assisted capacity region as well.

Before we present the next result, let $\mathcal{N}^{\semidet}$ be the set of 2-user BCs where User $2$ has a deterministic channel, i.e., $Y_2=f(X)$ is a function of $X$. This is known as the set of semi-deterministic broadcast channels (for which User $2$ has the deterministic channel, whereas the channel for User $1$ remains general). 
The following corollary is the characterization of the fully NS assisted capacity region for semi-deterministic broadcast channels.
\begin{corollary}[Semi-Det] \label{cor:semi_det}
  For a semi-deterministic broadcast channel $\mN^{\semidet} \in \mathcal{N}^{\semidet}$, we have
  \begin{align}
    &\mathcal{C}^{\NS} \big( \mN^{\semidet} \big)\notag \\
    &= \mathcal{R}_{\KS} \big( \mN^{\semidet} \big) \\
    &= \mathcal{R}_{\Sato} \big( \mN^{\semidet} \big) \\
    &= \bigcup_{\mP_X}
    \left\{
    \begin{array}{l}
      (R_1,R_2)\colon\\
      R_1\leq I(X;Y_1) \\
      R_2 \leq H(Y_2) \\
      R_1+R_2 \leq I(X;Y_1\mid Y_2) +  H(Y_2) 
    \end{array}
    \right.
  \end{align}
\end{corollary}
\noindent We are able to generalize the result beyond semi-deterministic BCs (Corollary \ref{cor:semi_det}) to a subset of semi-deterministic erasure BCs (Corollary \ref{cor:semi_det_E}) and to parallel reversely semi-deterministic BCs (Corollary \ref{cor:PRSD}). 

Let us first define the set of semi-deterministic erasure BCs. For $0\leq \gamma_1, \gamma_2\leq 1$, let $\mathcal{N}^{\semidetE}(\gamma_1,\gamma_2)$ be the set of semi-deterministic erasure broadcast channels, obtained by concatenating erasure channels, with erasure probability $\gamma_i$ to User $i$'s output of a semi-deterministic broadcast channel, for $i\in \{1,2\}$.  
Formally, a channel $\mN^{\semidetE}_{\tilde{Y}_1\tilde{Y}_2\mid X} \in \mathcal{N}^{\semidetE}(\gamma_1,\gamma_2) \subseteq \mathcal{P}((\mathcal{Y}_1\cup \{\phi\}) \times (\mathcal{Y}_2\cup \{\phi\}) \mid \mathcal{X})$ iff there exists $\mN^{\semidet} \in \mathcal{N}^{\semidet} \subseteq \mathcal{P}(\mathcal{Y}_1\times \mathcal{Y}_2 \mid \mathcal{X})$, $\phi\not\in \mathcal{X}\cup \mathcal{Y}_1\cup \mathcal{Y}_2$, such that for $i\in \{1,2\}$,
\begin{align}
  \mN^{\semidetE}_{\tilde{Y}_i\mid X}(\tilde{y}_i \mid x) = 
  \begin{cases}
    \mN_{Y_i\mid X}^{\semidet}(\tilde{y}_i  \mid x) \times (1-\gamma_i), & \tilde{y}_i \in \mathcal{Y}_i\\
    \gamma_i, & \tilde{y}_i  = \phi
  \end{cases}
\end{align}
where $\mN^{\semidet}_{Y_i\mid X}$ and $\mN^{\semidetE}_{Y_i\mid X}$ are the marginal distributions for User $i$ of channels $\mN^{\semidet}_{Y_i\mid X}$ and $\mN^{\semidetE}_{\tilde{Y}_i\mid X}$ , respectively.

\begin{corollary}[Semi-Det-Erasure] \label{cor:semi_det_E}
  For $\gamma_1\geq \gamma_2$, and for a channel $\mN^{\semidetE}_{\tilde{Y}_1\tilde{Y}_2\mid X}\in \mathcal{N}^{\semidetE}(\gamma_1,\gamma_2)$, obtained by concatenating the erasure channel with erasure probability $\gamma_i$ to the output $Y_i$ of the channel $\mN_{Y_1Y_2\mid X}^{\semidet} \in \mathcal{N}^{\semidet}$, for $i\in \{1,2\}$, we have
  \begin{align}
    &\mathcal{C}^{\NS}\big(\mN^{\semidetE}_{\tilde{Y}_1\tilde{Y}_2\mid X}\big) \notag \\
    &= \mathcal{R}_{\KS}\big(\mN^{\semidetE}_{\tilde{Y}_1\tilde{Y}_2\mid X}\big) \\
    &= \mathcal{R}_{\Sato}(\mN^{\semidetE}_{\tilde{Y}_1\tilde{Y}_2\mid X})\\
    &= \bigcup_{\mP_X}
    \left\{
    \begin{array}{l}
      (R_1,R_2)\colon\\
      R_1\leq (1-\gamma_1)\times I(X;Y_1) \\
      R_2 \leq (1-\gamma_2) \times H(Y_2) \\
      R_1+R_2 \leq  (1-\gamma_1)\times I(X;Y_1\mid Y_2) + (1-\gamma_2) \times H(Y_2) 
    \end{array}
    \right. \label{eq:semidetE_explicit}
  \end{align}
\end{corollary}
\noindent Theorem \ref{thm:bipartite_NSBC} and Theorem \ref{thm:Sato} already showed that $\mathcal{R}_{\KS}=C^{\NSa}\subseteq C^{\NS}\subseteq \mathcal{R}_{\Sato}$ for any channel. Therefore, the proof of Corollary \ref{cor:semi_det_E} can be completed by showing that $\mathcal{R}_{\KS}\big(\mN^{\semidetE}_{\tilde{Y}_1\tilde{Y}_2\mid X}\big)$ equals $\mathcal{R}_{\Sato}\big(\mN^{\semidetE}_{\tilde{Y}_1\tilde{Y}_2\mid X}\big)$, with both evaluating to \eqref{eq:semidetE_explicit}.
 The details appear in Appendix \ref{proof:semi_det_E}.

\begin{example}[BLEC]
  The Blackwell channel ($\mN_{\mathrm{\scriptscriptstyle BLC}}$) has input $X\in \{0,1,2\}$, outputs $(Y_1,Y_2) \in \{0,1\}^2$ such that when $X=0,Y_1 =Y_2 =0$; when $X=1,Y_1 =Y_2 =1$; when $X=2,Y_1 =0,Y_2 =1$. The classical capacity region of $\mN_{\mathrm{\scriptscriptstyle BLC}}$ is known to be the set of $(R_1,R_2)$ such that $R_1\leq H(Y_1), R_2\leq H(Y_2), R_1+R_2 \leq H(Y_1,Y_2)$ for some $\mP_X$. The Blackwell erasure channel ($\mN_{\mathrm{\scriptscriptstyle BLEC}}$) is defined by concatenating two erasure channels to $\mN_{\mathrm{\scriptscriptstyle BLC}}$ with the probability of erasure $\gamma_1=\gamma_2=\gamma$. Corollary \ref{cor:semi_det_E} implies that the NS assisted capacity region for $\mN_{\mathrm{\scriptscriptstyle BLEC}}$ is the classical capacity region of $\mN_{\mathrm{\scriptscriptstyle BLC}}$ scaled by $(1-\gamma)$.  Remarkably, it is known that the classical capacity region of $\mN_{\mathrm{\scriptscriptstyle BLEC}}$ is strictly smaller than $(1-\gamma)\times \mathcal{C}^{\C}(\mN_{\mathrm{\scriptscriptstyle BLC}})$ \cite[Lem. 2]{gohari2021outer}, for all $0<\gamma < 1$. This implies that NS assistance increases the capacity region of the Blackwell erasure channel, for all $0<\gamma < 1$.
\end{example}
 
Let us next define the set of parallel reversely semi-deterministic BCs. See Fig. \ref{fig:PRSD} for an illustration.
Let ${\sf N}'_{Y_1'Y_2'\mid X'}$ be a semi-deterministic BC where User $2$'s output $Y_2'$ is deterministic given $X'$, (i.e., $Y_2' = f(X')$), and let ${\sf N}''_{Y_1''Y_2''\mid X''}$ be another semi-deterministic BC where the output $Y_1''$ of User $1$  is deterministic given $X''$, (i.e., $Y_1'' = g(X'')$). Now let ${\sf N}_{Y_1Y_2\mid X}$ be a BC for which $X = (X',X'')$, $Y_1 = (Y_1',Y_1'')$ and $Y_2 = (Y_2',Y_2'')$, are obtained by placing ${\sf N}'_{Y_1'Y_2'\mid X'}$ and ${\sf N}''_{Y_1''Y_2''\mid X''}$ in parallel. Formally,
\begin{align}
  {\sf N}_{Y_1Y_2\mid X}(y_1, y_2\mid x) = {\sf N}'_{Y_1'Y_2'\mid X'}(y_1', y_2'\mid x') \times {\sf N}''_{Y_1''Y_2''\mid X''}(y_1'', y_2''\mid x''),
\end{align}
where $x=(x',x''), y_1 = (y_1',y_1''), y_2 = (y_2', y_2'')$.
We say that the BC ${\sf N}_{Y_1Y_2\mid X}$ is parallel reversely semi-deterministic.
\begin{figure}[htbp]
\centering
\begin{tikzpicture}
\node (Xp) at (0,0) {\small $X'$};
\node (Xpp) at (0,-2.5) {\small $X''$};
\node (Y1p) at (4,0) {\small $Y_1'$};
\node (Y1pp) at (4.75,-0.65) {\small $Y_1''=g(X'')$};
\node (Y2p) at (4.65,-1.85) {\small $Y_2'=f(X')$};
\node (Y2pp) at (4,-2.5) {\small $Y_2''$};
\draw[thick] (Xp) --  node[midway, sloped, draw, rectangle, thick, fill=white, inner sep=2pt, font=\scriptsize] {$\mN'_{Y_1'\mid X'}$} (Y1p);
\draw[thick] (Xpp) --  node[midway, sloped, draw, rectangle, thick, fill=white, inner sep=2pt, font=\scriptsize] {$\mN''_{Y_2''\mid X''}$} (Y2pp);
\draw[thick] (Xp) --  node[pos=0.4, sloped, draw, rectangle, thick, fill=white, inner sep=2pt, font=\small] {$f$} (Y2p.west);
\draw[thick] (Xpp) --  node[pos=0.4, sloped, draw, rectangle, thick, fill=white, inner sep=2pt, font=\small] {$g$} (Y1pp.west);

\draw[decorate, decoration={brace, amplitude=5pt}, thick] (-0.35,-2.65) -- (-0.35,0.1);
\node (X) at (-0.8,-1.25) {$X$};

\draw[decorate, decoration={brace, amplitude=5pt}, thick] (6,0.2) -- (6,-0.85);
\node (Y1) at (6.5,-0.35) {$Y_1$};

\draw[decorate, decoration={brace, amplitude=5pt}, thick] (6,-1.65) -- (6,-2.7);
\node (Y2) at (6.5,-2.2) {$Y_2$};
\end{tikzpicture}
\caption{Parallel reversely semi-deterministic BC. $X=(X',X'')$ is the input at the transmitter. $Y_k = (Y_k', Y_k'')$ is the output at User $k$, for $k\in \{1,2\}$.}
\label{fig:PRSD}
\end{figure}
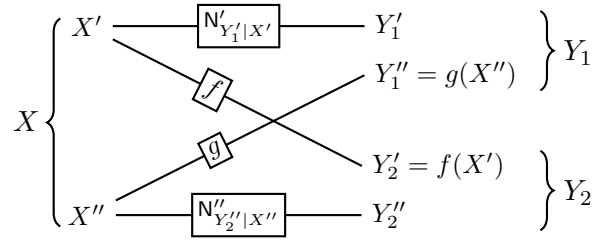

\begin{corollary}[Parallel Reversely Semi-Det] \label{cor:PRSD}
  For a parallel reversely semi-deterministic broadcast channel ${\sf N}_{Y_1Y_2\mid X}$, for which $X=(X',X''), Y_1 = (Y_1',Y_1''), Y_2=(Y_2',Y_2'')$ and $Y_2' = f(X')$, $Y_1''=g(X'')$, we have
  \begin{align}
    &\mathcal{C}^{\NS}\big( {\sf N}_{Y_1Y_2\mid X} \big)\notag \\
    &= \mathcal{R}_{\Sato} \big( {\sf N}_{Y_1Y_2\mid X} \big) \\
    &= \bigcup_{\mP_X = \mP_{X'}\times\mP_{X''}} 
    \left\{
    \begin{array}{l}
      (R_1,R_2)\colon\\
      R_1\leq   I(X';Y_1')+H(Y_1'') \\
      R_2 \leq  H(Y_2')+I(X'';Y_2'') \\
      R_1+R_2 \leq   H(Y_1'')+ H(Y_2') + I(X';Y_1'\mid Y_2') + I(X'';Y_2''\mid Y_1'')
    \end{array}
    \right.
  \end{align}
\end{corollary}
\noindent The proof of Corollary \ref{cor:PRSD} is done by observing that for a parallel reversely semi-deterministic BC, $\mathcal{R}_{\Sato}$ is achieved by independently applying Theorem \ref{thm:bipartite_NSBC} to ${\sf N}'_{Y_1'Y_2'\mid X'}$ and ${\sf N}''_{Y_1''Y_2''\mid X''}$ and taking their Minkowski sum, for every $\mP_{X'}$ and $\mP_{X''}$. We present the details in Appendix \ref{proof:PRSD}.

\section{Extension: NS-assisted $K$-user Broadcast with User Side information}
In this section, we apply the insights gained from previous sections to the $K$-user broadcast channels with user side information. A $K$-user broadcast channel is described by $\mN\in \mathcal{P}(\mathcal{Y}_1\times \mathcal{Y}_2\times \cdots \times \mathcal{Y}_K\mid \mathcal{X})$. For $k\in [K]$, let $W_k$ be an independent message for User $k$, and let $\mathcal{W}_k \subseteq [K]\setminus \{k\}$, so that for $k\in [K]$, User $k$ has $\{W_j\}_{j\in \mathcal{W}_k}$ as its side information.   In this section let us only consider full NS assistance among all $K+1$ parties (including the transmitter and $K$ users).

To have a concise notation, let $\mathcal{S}_k \triangleq \prod_{j\in \mathcal{W}_k}\mathcal{M}_j$ for $k\in [K]$.
Without loss of generality, assume $\mathcal{M}_k = [M_k]$, so that $M_k = |\mathcal{M}_k|$, for $k\in [K]$.  The set of fully NS assisted coding schemes with side information is defined as,
\begin{align}
  \mathcal{Z}^{\NS}(\mathrm{SI}=[K];M_1,\cdots, M_K) \triangleq \mathcal{P}^{\NS}_{K+1}
  {\footnotesize\left(
  \begin{matrix}
    \mathcal{X} \\ \mathcal{M}_1 \\ \mathcal{M}_2 \\ \vdots \\ \mathcal{M}_K
  \end{matrix}
  \left|
  \begin{matrix}
    \prod_{k=1}^K \mathcal{M}_k  \\
    \mathcal{Y}_1 \times \mathcal{S}_1 \\
    \mathcal{Y}_2 \times \mathcal{S}_2 \\
    \vdots \\
    \mathcal{Y}_K \times \mathcal{S}_K
  \end{matrix}
  \right.
  \right)}.
\end{align}
Note that we write the parameters of $\mathcal{P}^{\NS}_{K+1}$ in a way that each party appears in a separate row. The first parameter `$\mathrm{SI}=[K]$' indicates that all $K$ users have side information available. A coding scheme $\mZ\in \mathcal{Z}^{\NS}(\mathrm{SI}=[K];M_1,\cdots, M_K)$ works as follows. Let $\mathbf{W}=[W_1,W_2,\cdots, W_K]$, $S_k = (W_j)_{j\in \mathcal{W}_k}$ for $k\in [K]$. The transmitter inputs $\mathbf{W}$ to $\mZ$, and obtains the output $X$. This $X$ is sent through the channel $\mN$. For $k\in [K]$, User $k$ inputs $(Y_k, S_k)$ to $\mZ$, and obtains the decoded message $\hat{W}_k$.
Given $\mN$ and $\mZ\in \mathcal{Z}^{\NS}(\mathrm{SI}=[K];M_1,\cdots, M_K)$, the joint  distribution of all variables of interest is (using the `bold' notation, e.g.,  $\mathbf{A}$ to represent $[A_1,A_2,\cdots, A_K]$), 
\begin{align}
  &\Pr(\hat{\mathbf{W}} =\hat{\mathbf{w}}, {\mathbf{W}}=\mathbf{w},X=x,{\bf Y}={\bf y}) \notag \\
  &=\frac{1}{\prod_{k=1}^K M_k}  \Pr(\hat{\mathbf{W}}=\hat{\mathbf{w}}, X=x, \mathbf{Y}=\mathbf{y} \mid \mathbf{W}=\mathbf{w})\\
  &= \frac{1}{\prod_{k=1}^K M_k}   \mN(\mathbf{y}\mid x) \times \mZ {\footnotesize \left(
  \begin{matrix}
    x \\ \hat{w}_1 \\ \hat{w}_2 \\ \vdots \\ \hat{w}_K
  \end{matrix}
  \left|
  \begin{matrix}
     \mathbf{w}  \\
    [y_1,s_1] \\
    [y_2,s_2] \\
    \vdots \\
    [y_K,s_K]
  \end{matrix}
  \right.
  \right)}
\end{align}
where $s_k \triangleq (w_j)_{j\in \mathcal{W}_k} \in \mathcal{S}_k$ for $k\in [K]$.

We also consider the cases when the side information becomes unavailable at some users. To this end we accordingly define $\mathcal{Z}^{\NS}(\mathrm{SI}=\mathcal{K})$ (the parameters $M_1,\cdots, M_K$ are suppressed), for $\mathcal{K}\subseteq [K]$, indicating that the side information $S_k$ is only available for User $k\in \mathcal{K}$. Equivalently, the side information is removed from User $j$ for $j\in [K]\setminus \mathcal{K}$.  As an example,

\begin{align}
  \mathcal{Z}^{\NS}(\mathrm{SI}=\{2,\cdots, K\})\triangleq \mathcal{P}^{\NS}_{K+1}
  {\footnotesize \left(
  \begin{matrix}
    \mathcal{X} \\ \mathcal{M}_1 \\ \mathcal{M}_2 \\ \vdots \\ \mathcal{M}_K
  \end{matrix}
  \left|
  \begin{matrix}
    \prod_{k=1}^K \mathcal{M}_k  \\
    \mathcal{Y}_1  \\
    \mathcal{Y}_2 \times \mathcal{S}_2 \\
    \vdots \\
    \mathcal{Y}_K \times \mathcal{S}_K
  \end{matrix}
  \right.
  \right)}
\end{align}
and note that $\mZ\in \mathcal{Z}^{\NS}(\mathrm{SI}=\{2,\cdots, K\})$ does not allow an input of $S_1$ at User $1$.

Our main goal of this section is to study and compare the optimal probability of success corresponding to different side information availabilities at the users, assuming full NS assistance. To this end, let $\eta(\mZ)=\Pr(\hat{\mathbf{W}}=\mathbf{W})$ and for $\mathcal{K} \subseteq [K]$, let 
\begin{align}
  \eta_{\NS}^*(\mathrm{SI}=\mathcal{K}) \triangleq \max_{\mZ \in \mathcal{Z}^{\NS}(\mathrm{SI}=\mathcal{K})} \eta(\mZ)
\end{align}
be the optimal (joint) probability of success of the fully NS-assisted coding schemes with user side information only available to each User $k$, such that $k\in \mathcal{K}$. Note that this value also depends on $(M_1,\cdots, M_K)$ and the channel $\mN$ but those parameters are suppressed for compact notation. Certainly, since user side information cannot hurt the optimal probability of success, we have
\begin{align}
  \eta_{\NS}^*(\mathrm{SI}=\mathcal{K}) \geq \eta_{\NS}^*(\mathrm{SI}=\mathcal{K}')
\end{align}
as long as $\mathcal{K} \supseteq \mathcal{K}'$.

The main theorem of this section is the following.

\begin{theorem} \label{thm:extension}
  For $k\in [K]$, if $k\not\in  \mathcal{W}_1\cup \mathcal{W}_2 \cup \cdots \cup \mathcal{W}_K$, then
  \begin{align}
    \eta_{\NS}^*(\mathrm{SI}=[K]) = \eta_{\NS}^*(\mathrm{SI}=[K]\setminus \{k\}).
  \end{align}
\end{theorem}
\noindent The theorem says that the optimal probability of success of fully NS-assisted coding schemes does not change (decrease) if $S_k$ is removed from User $k$, provided that $W_k$ is not available as side information to any user. The proof is presented in Appendix \ref{proof:extension}. The next corollary follows directly from Theorem \ref{thm:extension}.
\begin{corollary} \label{cor:removing_SI}
  If $\mathcal{W}_k = \{k+1,k+2,\cdots, K\}$ for all $k\in [K]$, then
  \begin{align}
    \eta^*(\mathrm{SI}=\emptyset) = \eta^*(\mathrm{SI}=[K])
  \end{align}
  and therefore $\eta^*(\mathrm{SI}=\mathcal{K}) = \eta^*(\mathrm{SI}=\mathcal{K}')$ for every $\mathcal{K}, \mathcal{K}'\subseteq [K]$. 
\end{corollary}
The corollary says that if the side information structure is defined such that each User $k$ has $\{W_j\}_{j>k}$ as side information, for all $k\in [K]$, then such side information does not help the probability of success at all, with full NS assistance.
\begin{proof}
  Since $W_1$ is not known by any users, Theorem \ref{thm:extension} implies that removing its side information from User $1$ does not hurt the optimal probability of success. Note that User $1$ was the only user that knows $W_2$, so after this removal, $W_2$ is not known by any users. Then removing the side information from User $2$ does not hurt the optimal probability of success. Continue this argument until we remove the side information from User $K-1$. Since User $K$ does not have any side information initially, we obtain a setting where no user has side information, but the optimal probability of success (under full NS assistance) is not affected.
\end{proof}

The immediate consequence of Corollary \ref{cor:removing_SI} is that the (fully) NS assisted capacity region of a $K$-user broadcast channel does not change if in addition each User $k, k\in[K]$ is provided in advance all $\{W_j\}_{j>k}$. This equivalence may provide useful insights for  future studies of  NS assisted index coding \cite{Birk_Kol_Trans}, topological interference management \cite{Jafar_TIM}, linear computation broadcast \cite{Sun_Jafar_CBC, Yao_Jafar_3LCBC}, coded caching \cite{Maddah_Ali_Niesen}, and network coding problems \cite{Ahlswede_network_information_flow}.

\section{Conclusion}
The key insight of this work, encapsulated as `\emph{virtual CSIT signaling via NS assistance},'  emerged from our study of the NS-assisted capacity of a point to point channel with non-causal CSIT. Following this  insight, we found explicit computable NS-assisted capacity (region) expressions for a series of network communication problems. Let us conclude by pointing out a few promising directions for future work. First, an important open question is to determine if Sato's outer bound \emph{always} matches the fully NS assisted capacity region of a BC. Note that this possibility is supported by all available results thus far. Second, a promising avenue, based on Theorem \ref{thm:extension}, is to explore the NS assisted capacity of a BC with various forms of classical side information at each receiver, which includes the NS-assisted index coding problem as a most interesting special case. Third, another fundamental open problem is to characterize how the `virtual CSIT signaling' result generalizes if the message is not uniformly distributed, which would impact the twirling argument. Last but not the least, an important direction for future work is to use the strongest `metal detector' signals produced by NS-assisted capacity analysis, to search for the most significant capacity advantages achievable with \emph{quantum}-correlations.

\appendix
\section{Proof of Theorem  \ref{thm:csit_tp}} \label{proof:csit_tp}
We only need to show that $\max_{\mZ\in \mathcal{Z}^{\NS}(M,~\mN)}\eta(\mZ) \geq \max_{\mZ \in \mathcal{Z}^{\NS}(M,~\bar{\mN})}\eta(\mZ)$. 
Given any coding scheme $\mZ\in \mathcal{Z}^{\NS}(M, \bar{\mN})$, define
\begin{align}
  \mZ'(x,\hat{w}\mid [w,s], [y,s_{\tR}]) \triangleq \frac{1}{M}\sum_{\pi \in \mathbb{C}_M}\mZ(x, \pi(\hat{w})\mid [\pi(w),s], [y,s_{\tR}])
\end{align}
for $(w,\hat{w}, x,y,s,s_{\tR})\in \mathcal{M}^2\times \mathcal{X}\times \mathcal{Y} \times \mathcal{S}^2$,
and 
\begin{align}
    \mZ''(x,\hat{w}\mid [w,s],y) \triangleq \mZ'(x,\hat{w} \mid [w,s],[y,s]),
\end{align}
for $(w,\hat{w}, x,y,s)\in \mathcal{M}^2\times \mathcal{X}\times \mathcal{Y} \times \mathcal{S}$,
where $\mathbb{C}_M$ is the cyclic permutation group operating on $[M]$. 

We claim that $\mZ'$ and $\mZ''$ are non-signaling, and leave the proof to the end of this section.
Recall that 
\begin{align}
  \bar{\mN}_{YS_{\tR}\mid XS}(y,s_{\tR} \mid x,s) \triangleq \mN_{Y\mid XS}(y\mid x,s) \times \mathbb{I}(s_{\tR}=s).
\end{align}
In the following, we omit writing the domain of the variables as it is clear from the context.
By the definitions of the probability of success and $\bar{\mN}$,
\begin{align}
  \eta(\mZ) &= \frac{1}{M} \sum_{w, s, x,y,s_{\tR}}  \mP_S(s) \times \bar{\mN}_{YS_{\tR}\mid XS}(y,s_{\tR}\mid x,s) \times \mZ(x,w \mid [w,s], [y,s_{\tR}])  \\
  &= \frac{1}{M} \sum_{s , x,y}   \mP_S(s)\times \mN_{Y\mid XS}(y\mid x,s) \times  \sum_{w\in [M]} \mZ(x,w \mid [w,s], [y,s]) 
\end{align}
whereas
\begin{align}
  \eta(\mZ') &= \frac{1}{M} \sum_{w, s, x,y,s_{\tR}}  \mP_S(s) \times \bar{\mN}_{YS_{\tR}\mid XS}(y,s_{\tR}\mid x,s) \times \mZ'(x,w \mid [w,s], [y,s_{\tR}])  \\
  &= \frac{1}{M} \sum_{s,x,y}   \mP_S(s)\times \mN_{Y\mid XS}(y\mid x,s) \times  \sum_{w\in[M]} \mZ'(x,w \mid [w,s], [y,s]).
\end{align}
Now note that, for any $s,x,y$,
\begin{align}
  &\sum_{w\in[M]}\mZ'(x,w \mid [w,s], [y,s]) \notag \\
  &= \frac{1}{M}\sum_{w\in [M]} \sum_{\pi\in \mathbb{C}_M} \mZ(x,\pi(w) \mid [\pi(w),s],[y,s]) \\
  &= \frac{1}{M}\sum_{w\in [M]} \sum_{i\in [M]} \mZ(x, i \mid [i,s],[y,s]) \\
  &= \sum_{i\in [M]} \mZ(x, i \mid [i,s],[y,s])\\
  &= \sum_{w\in [M]} \mZ(x, w \mid [w,s],[y,s])
\end{align}
It follows that $\eta(\mZ) = \eta(\mZ')$. Again by definition,
\begin{align}
  \eta(\mZ'') &= \frac{1}{M} \sum_{s,x,y}\mP_S(s)\times \mN_{Y\mid XS}(y\mid x,s) \times \sum_{w\in [M]} \mZ''(x,w\mid [w,s],y) \\
  &=\frac{1}{M} \sum_{s,x,y} \mP_S(s)\times \mN_{Y\mid XS}(y\mid x,s) \times \sum_{w\in [M]} \mZ'(x,w\mid [w,s],[y,s])
\end{align}
and therefore $\eta(\mZ'') = \eta(\mZ')$. This shows that $\eta(\mZ'')=\eta(\mZ)$.

Let us now verify that $\mZ'$ and $\mZ''$ are non-signaling. It is not difficult to verify that $\mZ' \in \mathcal{P}(\mathcal{X}\times \mathcal{M}\mid \mathcal{M}\times \mathcal{S}\times \mathcal{Y}\times \mathcal{S})$ and $\mZ'' \in \mathcal{P}(\mathcal{X}\times \mathcal{M}\mid \mathcal{M}\times \mathcal{S}\times \mathcal{Y})$ are valid conditional distributions. Then, since $\mZ$ is non-signaling, 
\begin{align}
  \sum_{x}\mZ(x,\hat{w} \mid [w,s], [y,s_{\tR}]) \triangleq \mZ_{\tR}(\hat{w}\mid [y,s_{\tR}])
\end{align} 
is a function of only $(\hat{w},y,s_{\tR})$, and
\begin{align}
  \sum_{\hat{w}}\mZ(x,\hat{w} \mid [w,s], [y,s_{\tR}]) \triangleq \mZ_{\tT}(x\mid [w,s])
\end{align}
is a function of only $(x,w,s)$. We then have
\begin{align}
  &\sum_{x}\mZ'(x,\hat{w} \mid [w,s],[y,s_{\tR}])\notag \\
  &= \frac{1}{M}\sum_{\pi\in \mathbb{C}_M} \sum_{x}  \mZ(x, \pi(\hat{w}) \mid [\pi(w),s],[y,s_{\tR}]) \\
  &= \frac{1}{M}\sum_{\pi\in \mathbb{C}_M}  \mZ_{\tR}(\pi(\hat{w})\mid [y,s_{\tR}]) \\
  &= \frac{1}{M} \label{eq:csittp_1}
\end{align}
which is a constant, so it does not depend on $(w,s)$, and 
\begin{align}
  &\sum_{\hat{w}} \mZ'(x,\hat{w}\mid [w,s], [y,s_{\tR}]) \notag \\
  & = \frac{1}{M}\sum_{\pi\in \mathbb{C}_M} \sum_{\hat{w}}  \mZ(x, \pi(\hat{w}) \mid [\pi(w),s],[y,s_{\tR}]) \\
  & = \frac{1}{M}\sum_{\pi\in \mathbb{C}_M} \sum_{\hat{w}}  \mZ(x, \hat{w} \mid [\pi(w),s],[y,s_{\tR}]) \\
  & = \frac{1}{M}\sum_{\pi\in \mathbb{C}_M}   \mZ_{\tT}(x\mid [\pi(w),s]) \label{eq:csittp_2}
\end{align}
which does not depend on $(y,s_{\tR})$ (and also not on $w$ because of the summation over $\pi$). This proves that $\mZ'$ is non-signaling. Finally, let us note that \eqref{eq:csittp_1} implies
\begin{align}
  &\sum_{x}\mZ''(x,\hat{w}\mid [w,s], y) \notag \\
  &= \sum_{x}\mZ'(x,\hat{w} \mid [w,s],[y,s]) \\
  &= \frac{1}{M}
\end{align}
which does not depend on $(w,s)$ and that \eqref{eq:csittp_2} implies
\begin{align}
  &\sum_{\hat{w}}\mZ''(x,\hat{w}\mid [w,s], y) \notag \\
  &= \sum_{\hat{w}}\mZ'(x,\hat{w}\mid [w,s], [y,s]) \\
  &= \frac{1}{M}\sum_{\pi\in \mathbb{C}_M}   \mZ_{\tT}(x\mid [\pi(w),s])
\end{align}
which does not depend on $y$ (and $w$). This completes the proof that $\mZ'$ and $\mZ''$ are non-signaling. It follows that given $\mZ$, there exists a non-signaling $\mZ''$ which needs no input $S_{\tR}$ but still achieves the same probability of success, as desired. \hfill \qed

\section{Proof of Theorem \ref{thm:capacity_cws}} \label{proof:capacity_cws}
We shall show that $C^{\NS}(\mN) = \max_{\mP_{X\mid S}}I(X;Y\mid S)$, for a channel with state, $\mN=(\mN_{Y\mid XS}, \mP_S)$. The direction $C^{\NS}(\mN) \geq \max_{\mP_{X\mid S}}I(X;Y\mid S)$ is a direct implication of Theorem \ref{thm:csit_tp}, that $C^{\NS}(\mN)=C^{\NS}(\bar{\mN})$, and since classical coding schemes are contained in NS assisted coding schemes, we have $C^{\NS}(\bar{\mN}) \geq C^{\C}(\bar{\mN}) = \max_{\mP_{X\mid S}}I(X;Y\mid S)$.

To show the converse, that $C^{\NS}(\mN) \leq \max_{\mP_{X\mid S}}I(X;Y\mid S)$, we use a similar argument to \cite[Thm. 9]{matthews2012linear} to first show the following lemma.
\begin{lemma}\label{lem:NS_Fano}
For any NS assisted coding scheme $\mZ\in \mathcal{Z}^{\NS}(M,\mN)$,
  \begin{align}
  I(X;Y\mid S) \geq \eta(\mZ) \log_2(M) - H_b(\eta(\mZ)).
\end{align}
\end{lemma}
The proof is given in Appendix \ref{proof:NS_Fano}. Applying Lemma \ref{lem:NS_Fano} to $\mZ_n \in \mathcal{Z}^{\NS}(M_n,N^{\otimes n})$, we have that 
\begin{align}
  I(X^n;Y^n\mid S^n) \geq \eta(\mZ_n)\log_2(M_n) - H_b(\eta(\mZ_n)).
\end{align}
Then for any $R$ that is achievable by NS assistance,
\begin{align}
  R\leq \lim_{n\to \infty} \frac{1}{n}\log_2 (M_n) \leq  \lim_{n\to \infty} \frac{1}{n}I(X^n;Y^n\mid S^n) 
\end{align}
because $\lim_{n\to \infty} \eta(\mZ_n) = 1$.
The last step is to observe that
\begin{align} \label{eq:single_letter_cws}
  \frac{1}{n}I(X^n;Y^n\mid S^n)
  \leq \max_{\mP_{X\mid S}} I(X;Y\mid S).
\end{align}
Indeed, if \eqref{eq:single_letter_cws} could be violated, then  classical random coding over multiple blocks of $N^{\otimes \ell}$ for some $\ell >1$ would have exceeded the capacity of the channel with state when CSIT and CSIR are both available.
This completes the proof of the converse. \hfill \qed

\section{Proof of Theorem \ref{thm:bipartite_NSBC}} \label{proof:bipartite_NSBC}
Given a 2-user broadcast channel $\mN$, recall that $\mathcal{C}^{\NSa}(\mN)$ is the capacity region with bipartite NS assistance between the transmitter and User 1, and $\mathcal{C}^{\NSa,\SIa}(\mN)$ is the capacity region when, additionally, User 1 is given $W_2$.
Let us first show that $\mathcal{C}^{\NSa}(\mN) = \mathcal{C}^{\NSa,\SIa}(\mN)$. This is done by showing a stronger result, that
\begin{align} \label{eq:SI_equals_woSI}
  \max_{\mZ\in \mathcal{Z}^{\NSa}(M_1,M_2,\mN)} \eta(\mZ) = \max_{\mZ \in \mathcal{Z}^{\NSa,\SIa}(M_1,M_2,\mN)} \eta(\mZ).
\end{align}
This is saying that the joint probability of success cannot be improved even if $W_2$ is made available freely at User $1$, provided that NS assistance is available between the transmitter and User $1$. This is argued as follows. For $\mZ\in \mathcal{Z}^{\NSa,\SIa}(M_1,M_2,\mN)$, note that $\mZ=\mZ_1\times \mD_2$.
Construct $\mZ_1',\mZ_1''$ as
\begin{align}
  & \mZ_1'(x,\hat{w}_1\mid [w_1,w_2], [y_1,w_2^{\tR}])= \frac{1}{M_1}\sum_{\pi \in \mathbb{C}_{M_1}}\mZ_1(x,\pi(\hat{w}_1) \mid [\pi(w_1),w_2],[y_1,w_2^{\tR}]), \\
  & \mZ_1''(x,\hat{w}_1\mid [w_1,w_2],y_1) = \mZ_1'(x,\hat{w}_1\mid [w_1,w_2],[y_1,w_2]).
\end{align}
Following the same argument as in Appendix \ref{proof:csit_tp} by replacing $\hat{w}_1 \to \hat{w}, w_1\to w, w_2\to s, w_2^{\tR}\to s_{\tR}$, it can be shown that $\mZ_1'$ and $\mZ_1''$ are non-signaling. One can then similarly verify that $\eta(\mZ) = \eta(\mZ_1\times \mD)=\eta(\mZ_1'\times \mD) = \eta(\mZ_1''\times \mD)$, essentially because relabeling the message $W_1$ does not affect the (joint) probability of success, as $W_1$ is uniformly distributed.  The desired claim \eqref{eq:SI_equals_woSI} then follows as $\mZ_1''\times \mD \in \mathcal{Z}^{\NSa}(M_1,M_2,\mN)$.

It remains to show that $\mathcal{C}^{\NSa}(\mN) \subseteq \mathcal{R}_{\KS}(\mN)$. We need the following lemma.
\begin{lemma} \label{lem:NS_Fano_BC}
  For $\mZ\in \mathcal{Z}^{\NS}(M_1,M_2,\mN)$,
  \begin{align}
    &I(X;Y_1\mid W_2) \geq \eta_1(\mZ) \log_2(M_1) - H_b(\eta_1(\mZ)).
  \end{align}
\end{lemma}
The proof of Lemma \ref{lem:NS_Fano_BC} is left to Appendix \ref{proof:NS_Fano}. One may view this lemma as an application of Lemma \ref{lem:NS_Fano} by replacing $(W_1,W_2,Y_1)$ with $(W,S,Y)$.

The rest of the proof is to apply Lemma \ref{lem:NS_Fano_BC}, Fano's inequality, and the same single-letterization steps as in \cite{kramer2007capacity}, \cite[Sec. 5.6.1 and Thm. 8.5]{NIT} by identifying the relevant auxiliary random variables.

Applying Lemma \ref{lem:NS_Fano_BC} to $\mZ_n \in \mathcal{Z}^{\NS}(M_{1,n}, M_{2,n}, \mN^{\otimes n})$, we have 
\begin{align}
  I(X^n;Y_1^n\mid W_2) \geq \eta_1(\mZ_n) \log_2(M_{1,n}) - H_b(\eta_1(\mZ_n)).
\end{align}
Then for any $(R_1,R_2)$ that is achievable by fully NS assisted coding schemes,
\begin{align}
  R_1 \leq \lim_{n\to \infty} \frac{1}{n} \log_2(M_{1,n}) \leq \lim_{n\to \infty} \frac{1}{n} I(X^n;Y_1^n\mid W_2),
\end{align}
since $\lim_{n\to \infty} \eta_1(n) = 1$. We alternatively write it as
\begin{align} \label{eq:NSfano_R1}
  nR_1 \leq I(X^n;Y_1^n\mid W_2) + o(n).
\end{align}
\eqref{eq:NSfano_R1} also holds for $(R_1,R_2)\in \mathcal{C}^{\NSa}$, since $\mathcal{Z}^{\NSa} \subseteq \mathcal{Z}^{\NS}$.

Meanwhile, since there is no NS assistance to User $2$, $W_2\leftrightarrow Y_2^n \leftrightarrow \hat{W}_2$ form a Markov chain. By the data processing inequality and Fano's inequality,
\begin{align} \label{eq:fano_R2}
  n R_2\leq  I(W_2;Y_2^n) + o(n).
\end{align}

It is useful to have the joint distribution of $(W_1,W_2,X^n,Y_1^n, Y_2^n)$ written explicitly as,
\begin{align}
  &\Pr(W_1=w_1,W_2=w_2,X^n=x^n,Y_1^n=y_1^n,Y_2^n = y_2^n)\notag \\
  &=\frac{1}{M_1M_2}\mZ_{X^n\mid W_1W_2}(x^n\mid w_1,w_2)\times 
  \prod_{i=1}^n \mN(y_{1,i},y_{2,i}\mid x_i)
\end{align}
where $\mZ_{X^n\mid W_1W_2}$ is the marginal distribution of $\mZ_n$ for the transmitter.

Now, define $U_i \triangleq (W_2,Y_1^{i-1}, Y_{2,i+1}^n)$. Note that $U_i \leftrightarrow X_i \leftrightarrow (Y_{1,i},Y_{2,i})$ form a Markov chain for $i\in [n]$. We have
\begin{align}
  &I(X^n;Y_1^n\mid W_2) \notag \\
  &= \sum_{i=1}^n I(X^n;Y_{1,i}\mid W_2, Y_1^{i-1})\\
  &\leq \sum_{i=1}^n I(X^n,Y_{2,i+1}^n;Y_{1,i}\mid W_2, Y_1^{i-1})\\
  &= \sum_{i=1}^n I(Y_{2,i+1}^n;Y_{1,i}\mid W_2, Y_1^{i-1}) + \sum_{i=1}^n I(X^n;Y_{1,i} \mid \underbrace{W_2, Y_1^{i-1}, Y_{2,i+1}^n}_{= U_i})
\end{align}
and
\begin{align}
  &I(W_2;Y_2^n)\notag \\
  &=\sum_{i=1}^n I(W_2;Y_{2,i} \mid Y_{2,i+1}^n)  \\
  &=\sum_{i=1}^n I(W_2, Y_1^{i-1};Y_{2,i} \mid Y_{2,i+1}^n)  -   \sum_{i=1}^n I(Y_1^{i-1};Y_{2,i} \mid W_2, Y_{2,i+1}^n) \\
  &\leq \sum_{i=1}^n I(\underbrace{W_2, Y_1^{i-1},Y_{2,i+1}^n}_{= U_i};Y_{2,i})  -   \sum_{i=1}^n I(Y_1^{i-1};Y_{2,i} \mid W_2, Y_{2,i+1}^n) 
\end{align}
Therefore,
\begin{align}
  &I(X^n;Y_1^n\mid W_2)  + I(W_2;Y_2^n)   \notag \\
  &\leq  \sum_{i=1}^n I(X_i;Y_{1,i} \mid U_i) + \sum_{i=1}^n  I(U_i;Y_{2,i})  \label{eq:KornerMarton} \\
  &= n\big( I(X_T;Y_{1,T}\mid U_T,T) + I(U_T;Y_{2,T}\mid T) \big)\label{eq:time_sharing_variable}\\
  &\leq n\big( I(X_T;Y_{1,T}\mid U_T,T) + I(U_T,T;Y_{2,T})\big)
\end{align}
Step \eqref{eq:KornerMarton} is because $\sum_{i=1}^n I(Y_{2,i+1}^n;Y_{1,i}\mid W_2, Y_1^{i-1})=\sum_{i=1}^n I(Y_1^{i-1};Y_{2,i} \mid W_2, Y_{2,i+1}^n)$ by the Korner-Marton identity (Csisz{\'a}r sum identity), and $X^n \leftrightarrow (U_i,X_i) \leftrightarrow Y_{1,i}$. In \eqref{eq:time_sharing_variable} we define a time sharing variable $T$ that is uniformly distributed over $[n]$ and is independent of $W_1W_2X^nY_1^nY_2^n$. Note that $(U_T,T)\leftrightarrow X_T \leftrightarrow (Y_{1,T}, Y_{2,T})$.

On the other hand, from \eqref{eq:fano_R2} we have
\begin{align}
  &nR_2 -o(n) \notag \\
  &\leq I(W_2;Y_2^n) \\
  &= \sum_{i=1}^n I(W_2;Y_{2,i}\mid Y_{2,i+1}^n) \\
  &\leq \sum_{i=1}^n I(U_i;Y_{2,i})\\
  &\leq n\sum_{i=1}^n I(U_T,T;Y_{2,T})
\end{align}
Also, from \eqref{eq:NSfano_R1} we have
\begin{align}
  &nR_1 - o(n) \notag \\
  &\leq I(X^n;Y_1^n\mid W_2)\\
  &\leq I(X^n;Y_1^n) \label{eq:dropW1} \\
  &\leq \sum_{i=1}^n I(X_i;Y_{1,i}) \label{eq:XnYn2XiYi} \\
  &\leq n I(X_T;Y_{1,T})
\end{align}
where Step \eqref{eq:dropW1} is because $W_2\leftrightarrow X^n\leftrightarrow Y_1^n$, and Step \eqref{eq:XnYn2XiYi} is because $(X^n, Y_1^{i-1}) \leftrightarrow X_i \leftrightarrow Y_{1,i}$. By viewing $X_T$ as $X$, $Y_{1,T}$ as $Y_1$, $Y_{2,T}$ as $Y_2$, $(U_T,T)$ as $U$, we are able to show that any $(R_1,R_2)$ achievable with bipartite NS assistance between the transmitter and User 1 must satisfy
\begin{align}
  R_1 \leq I(X;Y_1), ~~R_2 \leq I(U;Y_2), ~~R_1+R_2 \leq I(X;Y_1\mid U)+I(U;Y_2)
\end{align}
for some $\mP_{XU}$ such that $U\leftrightarrow X \leftrightarrow (Y_1,Y_2)$. This proves that $\mathcal{C}^{\NSa} \subseteq \mathcal{R}_{\KS}$.\hfill \qed

\section{Proofs of Lemma \ref{lem:NS_Fano} and Lemma \ref{lem:NS_Fano_BC}} \label{proof:NS_Fano}
We first prove Lemma \ref{lem:NS_Fano}.
The proof is essentially the derivation of  \cite[Thm. 27 and Eq. (157)]{Polyanskiy_Poor_Verdu} (also see \cite[Thm. 9]{matthews2012linear}), by conditioning on each $S=s$ for $s\in \mathcal{S}$.
The key of the proof is identifying that $\mZ$ has probability of success equal to $1/M$ if the channel is broken,  conditioned on each $S=s$ for $s\in \mathcal{S}$. This observation is also useful for the generalization to broadcast channels, i.e., Lemma \ref{lem:NS_Fano_BC}. The rest of the proof then follows from an application of the data processing inequality.

Let $\eta(\mZ\mid s) \triangleq \Pr(W=\hat{W}\mid S=s) = \frac{1}{M}\sum_{w,x,y} \mN_{Y\mid XS}(y\mid x,s)\mZ(x,w \mid [w,s],y)$ be the probability of success conditioned on $S=s$, for $s\in \mathcal{S}$. Let
\begin{align}
T_{x,y}^s \triangleq \Pr(W=\hat{W}\mid X=x,Y=y,S=s)
\end{align}
for $s\in \mathcal{S}, x\in \mathcal{X}, y\in \mathcal{Y}$. By the law of total probability,
\begin{align} \label{eq:prob_success_closed}
  \eta (\mZ\mid s) = \sum_{x,y} \Pr(X=x,Y=y\mid S=s) \times T_{x,y}^s.
\end{align}
We point out that $T_{x,y}^s$ does not depend on $\mN_{Y|XS}$.  Indeed, $T_{x,y}^{s}$ is calculated as,
\begin{align}
  &T_{x,y}^s = \frac{\Pr(W=\hat{W}, X=x,Y=y\mid S=s)}{\Pr(X=x,Y=y\mid S=s)}\\
  &= \frac{\frac{1}{M}\sum_{w}\mZ(x,w \mid [w,s],y)N_{Y\mid XS}(y\mid x,s)}{\frac{1}{M}\sum_{w,\hat{w}}\mZ(x,\hat{w} \mid [w,s],y)N_{Y\mid XS}(y\mid x,s)}\\
  &=\frac{\sum_{w}\mZ(x,w \mid [w,s],y)}{\sum_{w,\hat{w}}\mZ(x,\hat{w} \mid [w,s],y)}
\end{align}
which does not depend on $\mN_{Y\mid XS}$.

On the other hand, we argue that
\begin{align} \label{eq:prob_success_broken}
  \sum_{x,y} \Pr(X=x\mid S=s) \times \Pr(Y=y\mid S=s) \times T_{x,y}^s = \frac{1}{M}.
\end{align}
This is because \eqref{eq:prob_success_broken} is equal to the probability of success of $\mZ$ if we replace $\mN_{Y\mid XS}$ with $\mP_{Y\mid S}$, i.e., break the channel, which makes $X$ and $Y$ independent conditioned on $S=s$, and then the probability of success must be equal to $1/M$ due to the fact that $\mZ$ is non-signaling. Indeed, conditioned on $S=s$, if we calculate the probability of success for $\mZ$ when the channel is broken, it equals
\begin{align}
  &\frac{1}{M}\sum_{w,x,y}\mZ(x,w\mid [w,s], y) \mP_{Y\mid S}(y\mid s) \\
  &=\frac{1}{M} \sum_{w,y} \mZ_{\tR}(w\mid y) \mP_{Y\mid S}(y\mid s)\\
  &= \frac{1}{M}
\end{align}
where we define $\sum_x \mZ(x,w \mid [w,s],y)\triangleq \mZ_{\tR}(w\mid y)$ as $\mZ$ is non-signaling.

Now for $s\in \mathcal{S}$, define a channel $\mB_s \in \mathcal{P}(\{0,1\}\mid \mathcal{X}\times \mathcal{Y})$ such that
$\mB_s(1\mid x,y)=T_{x,y}^s, \mB_s(0\mid x,y)=1-T_{x,y}^s$, for $(x,y) \in \mathcal{X}\times \mathcal{Y}$. In words, given that the input to the channel $\mB_s$ is $(x,y)$, the channel outputs $1$ with probability $T_{x,y}^s$ and outputs $0$ with probability $1-T_{x,y}^s$.  Let $\mB(\mP)\in \mathcal{P}(\{0,1\})$ denote the output distribution of channel $\mB$ with respect to input distribution $\mP$. Let $\mP_s, \mQ_s \in \mathcal{P}(\mathcal{X}\times\mathcal{Y})$ such that $\mP_s(x,y) \triangleq \Pr(X=x,Y=y\mid S=s)$ and $\mQ_s(x,y) \triangleq \Pr(X=x\mid S=s)\times \Pr(Y=y\mid S=s)$. Then \eqref{eq:prob_success_closed} implies that $\mB(\mP_s) = [1-\eta(\mZ\mid s), \eta(\mZ\mid s)]$ and \eqref{eq:prob_success_broken} implies that $\mB(\mathsf{Q}_s) = [1-\frac{1}{M}, \frac{1}{M}]$. It follows from the data processing inequality  that
\begin{align}
  I(X;Y\mid S=s) = D_{\mathrm{KL}}(\mP_s\Vert \mQ_s) \geq D_{\mathrm{KL}}(\mB(\mP_s) \Vert \mB(\mQ_s)) \geq \eta(\mZ\mid s)  \log_2(M) - H_b(\eta(\mZ\mid s)),
\end{align}
where $D_{\mathrm{KL}}(\cdot \Vert \cdot)$ is the KL divergence and $H_b(\cdot)$ is the binary entropy function. Taking the convex combination of both sides, we obtain
\begin{align}
  &I(X;Y\mid S) = \sum_{s} \mP_S(s) I(X;Y\mid S=s)\\
  &\geq \log_2 (M) \times \sum_{s} \mP_S(s) \eta(\mZ\mid s) -\sum_{s} \mP_S(s) H_b(\eta(\mZ\mid s))\\
  &= \eta(\mZ)  \log_2 (M) - \sum_{s} \mP_S(s) H_b(\eta(\mZ\mid s))\\
  &\geq \eta(\mZ)  \log_2 (M) - H_b(\eta(\mZ))
\end{align}
where the last step is because $H_b$ is concave. \hfill \qed

To prove Lemma \ref{lem:NS_Fano_BC}, we apply the same argument that was used to prove Lemma \ref{lem:NS_Fano}. We first show that, given a 2-user broadcast channel $\mN$ and a coding scheme $\mZ\in \mathcal{Z}^{\NS}(M_1,M_2,\mN)$, if the channel $\mN$ is replaced with the broken channel, $\mP_{Y_1Y_2}$, defined such that $\mP_{Y_1Y_2}(y_1,y_2)=\Pr(Y_1=y_1,Y_2=y_2)$, then conditioned on $W_2=w_2$, the probability of success for Message $1$ must equal $1/M_1$, for every $w_2\in [M_2]$. This is explicitly calculated as follows.
\begin{align}
  &\frac{1}{M_1} \sum_{w_1,\hat{w}_2,x,y_1,y_2} \mZ(x,w_1,\hat{w}_2\mid [w_1,w_2],y_1,y_2)\times \mP_{Y_1Y_2}(y_1,y_2)\\
  &=\frac{1}{M_1} \sum_{w_1,\hat{w}_2,y_1,y_2} \mZ_{\tR_\mathrm{\scriptscriptstyle  1}\tR_\mathrm{\scriptscriptstyle  2}}(w_1,\hat{w}_2\mid y_1,y_2)\times \mP_{Y_1Y_2}(y_1,y_2)\\
  &=\frac{1}{M_1}
\end{align}
where $\mZ_{\tR_\mathrm{\scriptscriptstyle  1}\tR_\mathrm{\scriptscriptstyle 2}}$ is the marginal distribution of $\mZ$ for User $1$ and User $2$. Then by viewing $W_2$ as $S$, $Y_1$ as $Y$, we conclude that $I(X;Y_1\mid W_2) \geq \eta_1(\mZ)\log_2(M_1) - H_b(\eta_1(\mZ))$ as in the proof of Lemma \ref{lem:NS_Fano}. \hfill \qed

\section{Proof of Theorem \ref{thm:Sato}} \label{proof:Sato}
The proof is simply a combination of the converse argument for the NS assisted capacity of a point-to-point channel (without states) and the same-marginals property of NS assisted coding schemes for general broadcast channels (see \cite[Thm. 2]{Yao_Jafar_NS_DoF}).

First let us recall the converse for the NS assisted capacity of a point-to-point channel (without states). This can be deduced from a special case of Lemma \ref{lem:NS_Fano} by making $|\mathcal{S}|=1$ (i.e., no channel state $S$), and applying it to $n$-th extension of a point-to-point channel $\mN_{Y\mid X}^{\otimes n}$. We obtain $I(X^n;Y^n) \geq \eta(\mZ_n) \log_2(M_n) - H_b(\eta(\mZ_n))$, for $\mZ_n\in \mathcal{Z}^{\NS}(M_n,N_{Y\mid X}^{\otimes n})$, and therefore $nR\leq I(X^n;Y^n)+o(n)$ if $R$ is achievable by NS assistance over the channel $\mN_{Y\mid X}$.

Now consider the broadcast channel. 
Given $\mN_{Y_1Y_2\mid X} \in \mathcal{P}(\mathcal{Y}_1 \times \mathcal{Y}_2 \mid \mathcal{X})$, let $\mN_{Y_1\mid X}$ and $\mN_{Y_2\mid X}$ be the marginal distributions of $\mN_{Y_1Y_2\mid X}$ for User $1$ and User $2$, respectively.
Let $\mathcal{N}'(\mN_{Y_1Y_2\mid X})$ be the subset of channels that has the same marginals as $\mN_{Y_1Y_2\mid X}$, i.e., $\mN'_{Y_1'Y_2'}\in \mathcal{N}'(\mN_{Y_1Y_2\mid X})$  iff $\sum_{y_i}\mN'_{Y_1'Y_2'}(y_1,y_2\mid x) = \sum_{y_i}\mN_{Y_1Y_2\mid X}(y_1,y_2\mid x)$, for all $x\in \mathcal{X}, y_j\in\mathcal{Y}_j$, $\{i,j\}=\{1,2\}$.

Now say we are given a sequence of fully NS assisted coding schemes $\mZ_n \in \mathcal{Z}^{\NS}(M_{1,n},M_{2,n},$ $\mN_{Y_1Y_2\mid X}^{\otimes n})$. Denote $\mZ_{\tT\tR_\mathrm{\scriptscriptstyle  1},n}$ as the marginal distribution of $\mZ_n$ for the transmitter and User 1. It is not difficult to see that $\mZ_{\tT\tR_\mathrm{\scriptscriptstyle  1},n} \in \mathcal{Z}^{\NS}(M_{1,n},\mN_{Y_1\mid X}^{\otimes n})$ and therefore it is a valid NS assisted coding scheme over the sub-channel from the transmitter to User $1$. If $(R_1,R_2)$ is achievable by NS assisted coding schemes, it follows that $nR_1\leq I(X^n;Y_1^n) + o(n)$. By symmetry, we obtain that $nR_2\leq I(X^n;Y_2^n) + o(n)$. 

Similarly, if we treat User $1$ and User $2$ as one receiver by grouping their message $W_1,W_2$ as one message $W$, and letting $Y_1,Y_2$ to be processed together, then the argument also implies $n(R_1+R_2) \leq I(X^n;Y_1^n,Y_2^n)+o(n)$. Now according to \cite[Thm. 2]{Yao_Jafar_NS_DoF}, $\mZ_n$ achieves the same (individual) probability of success on any channel ${\mN'}_{Y_1'Y_2'\mid X}^{\otimes n}$ as long as $\mN'_{Y_1'Y_2'\mid X}\in \mathcal{N}'(\mN_{Y_1Y_2\mid X})$. It follows that, for $k\in \{1,2\}$,
\begin{align}
  &nR_k-o(n) \notag \\
  &\leq I(X^n;{Y_k'}^n)   \\
  &= I(X^n;Y_k^n)\\
  & = \sum_{i=1}^n I(X^n; Y_{k,i} \mid Y_k^{i-1})     \\
  &\stackrel{(a)}{\leq} \sum_{i=1}^n I(X_i;Y_{k,i})\\
  &=nI(X_T;Y_{k,T}\mid T)\\
  &\stackrel{(b)}{\leq} nI(X_T;Y_{k,T})
\end{align}
and
\begin{align}
  &n(R_1+R_2)-o(n) \notag \\
  &\leq \min_{\mN'_{Y_1'Y_2'\mid X}}I(X^n; {Y_1'}^n,{Y_2'}^n)\\
  &= \min_{\mN'_{Y_1'Y_2'\mid X}} \sum_{i=1}^nI(X^n;  Y_{1,i}' , Y_{2,i}'  \mid {Y_1'}^{i-1},{Y_2'}^{i-1})\\
  &\stackrel{(a)}{\leq} \min_{\mN'_{Y_1'Y_2'\mid X}}\sum_{i=1}^n I(X_i; Y'_{1,i}, Y'_{2,i})\\
  &=n \times \min_{\mN'_{Y_1'Y_2'\mid X}}  I(X_T; Y'_{1,T}, Y'_{2,T}\mid T) \\
  &\stackrel{(b)}{\leq} n \times \min_{\mN'_{Y_1'Y_2'\mid X}}I(X_T; Y'_{1,T}, Y'_{2,T})
\end{align}
where $T$ is uniformly distributed over $[n]$ and is independent of $W_1W_2X^n{Y_1'}^n{Y_2'}^n$. Steps labeled by $(a)$ follow from the Markov chain $(X^n,{Y'_1}^{i-1},{Y'_2}^{i-1}) \leftrightarrow X_i\leftrightarrow (Y_{1,i}',Y_{2,i}')$, and steps labeled by $(b)$ follow from the Markov chain $T\leftrightarrow X_T \leftrightarrow (Y'_{1,T},Y'_{2,T})$. We have also used the fact that if $X\leftrightarrow Y \leftrightarrow Z$ then $I(Y;Z\mid X)\leq I(Y;Z)$.
Thus, we obtain $(R_1,R_2) \in \mathcal{R}_{\Sato}(\mN_{Y_1Y_2\mid X})$,  by identifying $X= X_T, Y_1=Y_{1,T}, Y_2=Y_{2,T}, Y_1'=Y_{1,T}', Y_2'=Y_{2,T}'$.  \hfill \qed

\section{Proof of Corollary \ref{cor:semi_det_E}} \label{proof:semi_det_E}
The following lemma will be useful.
\begin{lemma}\label{lem:erasure_channel}
Consider any channel $\mN\in \mathcal{P}(\mathcal{Y}\mid \mathcal{X})$, and define the erasure symbol $\phi\notin\mathcal{X}\cup\mathcal{Y}$. Let $\mN_{\rm E}^\gamma \in \mathcal{P}(\mathcal{Y}\cup\{\phi\}\mid \mathcal{Y}\cup\{\phi\})$, denote the erasure channel with erasure probability $\gamma$, i.e., 
\begin{align}
\mN_{\rm E}^\gamma({y}_2 \mid y_1)&=\left\{\begin{array}{ll}(1-\gamma)\mathbb{I}(y_2=y_1)+\gamma\mathbb{I}(y_2=\phi), &y_1\not= \phi\\
\mathbb{I}(y_2=\phi),&y_1=\phi.
\end{array}
\right.
\end{align}
Consider $X\stackrel{\mN}{\longrightarrow} Y \stackrel{\mN_{\rm E}^{\alpha}}{\longrightarrow} \tilde{Y}' \stackrel{\mN_{\rm E}^{\beta}}{\longrightarrow} \tilde{Y}$, where $A \stackrel{\mN}{\longrightarrow} B$ means that $A,B$ are the input and the output of channel $\mN$, respectively. Then,
\begin{align}
  I(X;\tilde{Y}') = (1-\alpha)I(X;Y)
\end{align}
and
\begin{align}
  I(X;\tilde{Y}) = (1-\alpha)(1-\beta)I(X;Y)= (1-\beta) I(X;\tilde{Y}')
\end{align}
\end{lemma}
\begin{proof}
  Without loss of generality, let us map $\mathcal{X}\mapsto \{1,2,\cdots, |\mathcal{X}|\}$, $\mathcal{Y} \mapsto \{1,2,\cdots, |\mathcal{Y}|\}$, and $\phi \mapsto 0$. Then $\tilde{Y}' = E_{\alpha} \times Y$, $\tilde{Y}=E_{\beta}\times \tilde{Y}'$, where $E_\alpha, E_\beta$ are independent random variables that are also independent of $(X,Y)$ such that $[\Pr(E_x=0),\Pr(E_x=1)]=[x, 1-x]$ for $x\in \{\alpha, \beta\}$. Now
  \begin{align}
    &I(X;\tilde{Y}')= I(X;\tilde{Y}',E_\alpha)\\
    &=I(X;\tilde{Y}'\mid E_\alpha) \\
    &=(1-\alpha)I(X;\tilde{Y}' \mid E_\alpha=1)\\
    &=(1-\alpha) I(X;Y \mid E_\alpha=1)\\
    &=(1-\alpha) I(X;Y)
  \end{align}
  Similarly, 
  \begin{align}
    &I(X;\tilde{Y})= I(X;\tilde{Y}, E_\alpha\times E_\beta)\\
    &=(1-\alpha)(1-\beta)I(X;\tilde{Y}\mid E_\alpha \times E_\beta = 1)\\
    &= (1-\alpha)(1-\beta) I(X;Y)\\
    &= (1-\beta)I(X;\tilde{Y}')
  \end{align}
  which concludes the proof of the lemma.
\end{proof}
According to Theorem \ref{thm:bipartite_NSBC}, with $U = f(X)=Y_2$, it follows that any $(R_1,R_2)$ is achievable by full NS assistance if
\begin{align}
  R_1 &\leq I(X;\tilde{Y}_1) \\
  &\stackrel{(a)}{=} (1-\gamma_1) I(X;Y_1)
\end{align}
\begin{align}
  R_2 &\leq I(Y_2;\tilde{Y}_2) \\
  &= I(X;\tilde{Y}_2) \\
  &\stackrel{(a)}{=} (1-\gamma_2) I(X;Y_2) \\
  &= (1-\gamma_2) H(Y_2)
\end{align}
and
\begin{align}
  &R_1+R_2 \notag \\
  &\leq I(X;\tilde{Y}_1\mid Y_2) + I(Y_2;\tilde{Y}_2)\\
  &=I(X;\tilde{Y}_1\mid Y_2) + I(X;\tilde{Y}_2)\\
  &= \sum_{y_2\in \mathcal{Y}_2}\Pr(Y_2=y_2) \times I(X;\tilde{Y}_1\mid Y_2=y_2) + I(X;\tilde{Y}_2) \\
  &\stackrel{(a)}{=} (1-\gamma_1) I(X;Y_1\mid Y_2) + (1-\gamma_2) I(X;Y_2)\\
  &= (1-\gamma_1) I(X;Y_1\mid Y_2) + (1-\gamma_2) H(Y_2)
\end{align}
for some $P_X\in \mathcal{P}(\mathcal{X})$.
Steps labeled by $(a)$ follow from Lemma \ref{lem:erasure_channel}. This provides an inner bound on $\mathcal{C}^{\NS}\big(\mN^{\semidetE}_{\tilde{Y}_1\tilde{Y}_2\mid X}\big)$.

On the other hand, the sub-channel to User 1, $\mN^{\semidetE}_{Y_1\mid X}$ can be equivalently viewed as $\mN^{\semidet}_{Y_1\mid X}$ followed by a sequence of two erasure channels, with erasure probabilities $\gamma'$ and $\gamma_2$ such that $1-\gamma_1= (1-\gamma')(1-\gamma_2)$ (recall that we need $0\leq \gamma_2 \leq \gamma_1 \leq 1)$. Denote $\tilde{Y}_1'$ as the output of the first erasure channel (the one associated with $\gamma'$). Note that in computing $\mathcal{R}_{\Sato}\big(\mN^{\semidetE}_{\tilde{Y}_1\tilde{Y}_2\mid X}\big)$, we are free to choose the joint distribution subject to the same-marginals conditions, and any choice provides a valid outer bound on $\mathcal{C}^{\NS}\big(\mN^{\semidetE}_{\tilde{Y}_1\tilde{Y}_2\mid X}\big)$. Let us choose the joint distributions such that the erasure at User $2$ happens simultaneously with the second erasure at User $1$ (both erasure channels are associated with $\gamma_2$), so that $[\tilde{Y}_1', Y_2]$ undergoes a block erasure channel with erasure probability $\gamma_2$.
The outer bound then implies that any achievable $(R_1,R_2)$ must satisfy
\begin{align}
  R_1 &\leq I(X;\tilde{Y}_1) \\
  &\stackrel{(a)}{=} (1-\gamma_1) I(X;Y_1)
\end{align}
\begin{align}
  R_2 &\leq I(X;\tilde{Y}_2) \\
  &\stackrel{(a)}{=} (1-\gamma_2)  I(X;Y_2) \\
  &= (1-\gamma_2)  H(Y_2)
\end{align}
and
\begin{align}
  &R_1+R_2 \notag \\
  &\stackrel{(a)}{\leq} (1-\gamma_2) I(X;\tilde{Y}_1',Y_2)\\
  &=(1-\gamma_2)  I(X;Y_2) + (1-\gamma_2) I(X;\tilde{Y}_1'\mid Y_2)\\
  &=(1-\gamma_2) H(Y_2) + (1-\gamma_2) \sum_{y_2\in \mathcal{Y}_2}\Pr(Y_2=y_2)\times I(X;\tilde{Y}_1'\mid Y_2=y_2)\\
  &\stackrel{(a)}{=} (1-\gamma_2) H(Y_2) + (1-\gamma_2)(1-\gamma') I(X;Y_1\mid Y_2)\\
  &= (1-\gamma_1) I(X;Y_1\mid Y_2) + (1-\gamma_2)  H(Y_2)
\end{align}
for some $\mP_X\in \mathcal{P}(\mathcal{X})$. Steps labeled by $(a)$ follow from Lemma \ref{lem:erasure_channel}. Note that the outer bound matches the inner bound, which completes the proof. \hfill \qed

\section{Proof of Corollary \ref{cor:PRSD}} \label{proof:PRSD}
To apply Sato's outer bound (Theorem \ref{thm:Sato}) for the channel ${\sf N}_{Y_1Y_2\mid X}$, let us first note that it suffices to consider $\mP_{X'X''} = \mP_{X'}\times \mP_{X''}$ for which $X'$ and $X''$ are independent, due to the fact that the product distribution maximizes the mutual information over a product channel (cf. \cite[Problem 7.5]{Cover_book}). It follows that $\mathcal{C}^{\NS}\big({\sf N}_{Y_1Y_2\mid X}\big)$ is contained in the set of $(R_1,R_2)$ such that
\begin{align}
  R_1 &\leq   I(X';Y_1') + H(Y_1''), \label{eq:PRSD_outer_region_first} \\
  R_2 &\leq   H(Y_2') + I(X'';Y_2''), \\
  R_1+R_2 &\leq  I(X';Y_1'\mid Y_2') + H(Y_2')  + H(Y_1'')+ I(X'';Y_2''\mid Y_1'') \label{eq:PRSD_outer_region_end}
\end{align}
for some $\mP_{X'X''} = \mP_{X'}\times \mP_{X''}$, as $Y_2' = f(X')$, $Y_1''=g(X'')$.

In the following, we show that this region is achievable with full NS assistance by independently operating over ${\sf N}'_{Y_1'Y_2'\mid X'}$ and ${\sf N}''_{Y_1''Y_2''\mid X''}$. First, Theorem \ref{thm:bipartite_NSBC} shows that (by letting $U = f(X')$), $(R_1', R_2')$ is achievable with NS assistance for the BC ${\sf N}'_{Y_1'Y_2'\mid X'}$ if
\begin{align}
  R_1' &\leq I(X';Y_1') \label{eq:PRSD_region1_first} \\
  R_2' &\leq H(Y_2') \\
  R_1'+R_2' &\leq I(X';Y_1'\mid Y_2') + H(Y_2') \label{eq:PRSD_region1_end}
\end{align}
for some $\mP_{X'}$. By switching the indices of the two users, Theorem \ref{thm:bipartite_NSBC}  also shows that (by letting $U = g(X'')$), $(R_1'', R_2'')$ is achievable with NS assistance for the BC ${\sf N}''_{Y_1''Y_2''\mid X''}$ if
\begin{align}
  R_1'' &\leq H(Y_1'') \label{eq:PRSD_region2_first}\\
  R_2'' &\leq I(X'';Y_2'')\\
  R_1''+R_2'' &\leq H(Y_1'') + I(X'';Y_2''\mid Y_1'') \label{eq:PRSD_region2_end}
\end{align}
for some $\mP_{X''}$. Then, let $\mathcal{R}'(\mP_{X'})$ denote the region of $(R_1', R_2')$ specified by \eqref{eq:PRSD_region1_first} to \eqref{eq:PRSD_region1_end} given $\mP_{X'}$, and let $\mathcal{R}''(\mP_{X''})$ denote the region of $(R_1'', R_2'')$ specified by \eqref{eq:PRSD_region2_first} to \eqref{eq:PRSD_region2_end} given $\mP_{X''}$. 
It follows that $(R_1, R_2)$ is achievable for the BC ${\sf N}_{Y_1Y_2\mid X}$ if $R_1= R_1'+R_1''$, $R_2=R_2'+R_2''$, $(R_1',R_2')\in \mathcal{R}'(\mP_{X'})$, $(R_1'',R_2'')\in \mathcal{R}''(\mP_{X''})$ for some $\mP_{X'}$ and $\mP_{X''}$. This can be thought of as splitting the message $W_k$ into $(W'_k,W''_k)$ for $k\in \{1,2\}$, sending $(W_1', W_2')$ through ${\sf N}'_{Y_1'Y_2'\mid X'}$, and sending $(W_1'', W_2'')$ through ${\sf N}''_{Y_1''Y_2''\mid X''}$. Note that given $\mP_{X'}$ and $\mP_{X''}$, the region of $(R_1,R_2)$ thus achieved is the Minkowski sum $\mathcal{R}'(\mP_{X'})+\mathcal{R}''(\mP_{X''})$. Since $\mathcal{R}'(\mP_{X'})$ and $\mathcal{R}''(\mP_{X''})$ are polymatroids,\footnote{It suffices to check that the bound for $R_1'+R_2'$ is not greater than the sum, of the bound for $R_1'$, and the bound for $R_2'$, and similarly the bound for $R_1''+R_2''$ is not greater than the sum, of the bound for $R_1''$, and the bound for $R_2''$.} their Minkowski sum is obtained as the set of $(R_1,R_2)$ such that
\begin{align}
  R_1 &\leq I(X';Y_1')+H(Y_1')\\
  R_2 &\leq H(Y_2')+ I(X'';Y_2'')\\
  R_1+R_2 &\leq I(X';Y_1'\mid Y_2') + H(Y_2') + H(Y_1'') + I(X'';Y_2''\mid Y_1'')
\end{align}
which is again a polymatroid by adding the bounds individually (see e.g., \cite[Thm. 3]{mcdiarmid1975rado}). Comparing with the bounds from \eqref{eq:PRSD_outer_region_first} to \eqref{eq:PRSD_outer_region_end}, we conclude that $\mathcal{R}_{\Sato}\big( {\sf N}_{Y_1Y_2\mid X} \big)$ is achievable (with full NS assistance), thus equal to the fully NS assisted capacity region $\mathcal{C}^{\NS}\big( {\sf N}_{Y_1Y_2\mid X} \big)$. \hfill \qed

\section{Proof of Theorem \ref{thm:extension}} \label{proof:extension}
Without loss of generality, say $k=1$. We shall show that, if $W_1$ is not known by any users, then removing the side information of User $1$ does not affect the optimal probability of success (assuming full NS assistance). Given an NS assisted coding scheme where the side information is available at all users, i.e., $\mZ  \in \mathcal{Z}^{\NS}(\mathrm{SI}=[K])$, let
\begin{align}
  \mZ' \in \mathcal{P}\Bigg(\mathcal{X} \times \prod_{k=1}^K\mathcal{M}_k ~\Bigg|~ \prod_{k=1}^K(\mathcal{M}_k\times \mathcal{Y}_k \times \mathcal{S}_k)  \Bigg)
\end{align}
such that
\begin{align}
  \mZ'{\footnotesize \left(
  \begin{matrix}
    x\\ \hat{w}_1 \\ \hat{w}_2 \\ \vdots \\ \hat{w}_K
  \end{matrix} \left|  
  \begin{matrix}
    w_1,w_2,\cdots, w_K \\
    y_1,s_1^{\tR}\\
    y_2,s_2^{\tR}\\
    \vdots \\
    y_K,s_K^{\tR}
  \end{matrix}
  \right.
  \right)}
  \triangleq
  \frac{1}{M_1}\sum_{\pi \in \mathbb{C}_{M_1}}\mZ {\footnotesize \left(
  \begin{matrix}
    x\\ \pi(\hat{w}_1) \\ \hat{w}_2 \\ \vdots \\ \hat{w}_K
  \end{matrix} \left|  
  \begin{matrix}
    \pi(w_1),w_2,\cdots, w_K \\
    y_1,s_1^{\tR}\\
    y_2,s_2^{\tR}\\
    \vdots \\
    y_K,s_K^{\tR}
  \end{matrix}
  \right.
  \right)}
\end{align}
for all $x\in \mathcal{X}$, $\hat{w}_k\in \mathcal{M}_k, w_k \in \mathcal{M}_k, y_k \in \mathcal{Y}_k, s_k^{\tR} \in \mathcal{S}_k, k\in [K]$, where $\mathbb{C}_{M_1}$ is the cyclic permutation group operating on $[M_1]$. It is not difficult to verify that $\mZ'\in \mathcal{Z}^{\NS}(\mathrm{SI}=[K])$.
We then show that $\eta(\mZ) = \eta(\mZ')$, meaning that twirling $W_1$ does not affect the probability of success. This is because $W_1$ is uniformly distributed and no user has $W_1$. Specifically, recall the concise notation $s_k\triangleq (w_j)_{j\in \mathcal{W}_k}$. Then,
\begin{align}
  \eta(\mZ')&=\frac{1}{\prod_{k=1}^KM_k} \sum_{x, \mathbf{y}}\mN(\mathbf{y}\mid x) \sum_{w_2,\cdots, w_K} \frac{1}{M_1}\sum_{w_1}   \sum_{\pi\in \mathbb{C}_{M_1}}\mZ {\footnotesize \left(
  \begin{matrix}
    x \\ \pi(w_1) \\ w_2 \\ \vdots \\ w_K
  \end{matrix}
  \left|
  \begin{matrix}
     \pi(w_1),w_2,\cdots, w_K  \\
     y_1,s_1 \\
     y_2,s_2 \\
     \vdots \\
     y_K,s_K
  \end{matrix}
  \right.
  \right)}\\
  &=\frac{1}{\prod_{k=1}^KM_k} \sum_{x, \mathbf{y}}\mN(\mathbf{y}\mid x) \sum_{w_2,\cdots, w_K}  \sum_{w_1}    \mZ {\footnotesize \left(
  \begin{matrix}
    x \\ w_1 \\ w_2 \\ \vdots \\ w_K
  \end{matrix}
  \left|
  \begin{matrix}
     w_1,w_2,\cdots, w_K  \\
     y_1,s_1 \\
     y_2,s_2 \\
     \vdots \\
     y_K,s_K
  \end{matrix}
  \right.
  \right)}\\
  &= \eta(\mZ)
\end{align}

Now let 
\begin{align}
  \mZ'' \in \mathcal{P}\Bigg(\mathcal{X} \times \prod_{k=1}^K\mathcal{M}_k ~\Bigg|~ \prod_{k=1}^K(\mathcal{M}_k\times \mathcal{Y}_k) \times \prod_{j=2}^K \mathcal{S}_j \Bigg)
\end{align}
such that
\begin{align} \label{eq:def_Zpp_extention}
  \mZ''{\footnotesize \left(
  \begin{matrix}
    x \\ \hat{w}_1 \\ \hat{w}_2 \\ \vdots \\ \hat{w}_K
  \end{matrix}
  \left|
  \begin{matrix}
     w_1,w_2,\cdots, w_K  \\
     y_1\\
     y_2,s_2^{\tR} \\
     \vdots \\
     y_K,s_K^{\tR}
  \end{matrix}
  \right.
  \right)}
  \triangleq
  \mZ'{\footnotesize \left(
  \begin{matrix}
    x \\ \hat{w}_1 \\ \hat{w}_2 \\ \vdots \\ \hat{w}_K
  \end{matrix}
  \left|
  \begin{matrix}
     w_1,w_2,\cdots, w_K  \\
     y_1,s_1 \\
     y_2,s_2^{\tR} \\
     \vdots \\
     y_K,s_K^{\tR}
  \end{matrix}
  \right.
  \right)}
\end{align}
for all $x\in \mathcal{X}$, $\hat{w}_k\in \mathcal{M}_k, w_k \in \mathcal{M}_k, y_k \in \mathcal{Y}_k, k\in [K], s_j^{\tR} \in \mathcal{S}_j, j\in \{2,3,\cdots, K\}$.
We claim that $\mZ'' \in \mathcal{Z}^{\NS}(\mathrm{SI}=[K]\setminus \{1\})$, meaning that $\mZ''$ is also a valid NS assisted coding scheme, for which User $1$ does not input its side information. To show it, we need to show that $\mZ''$ satisfies in total $K+1$ conditions, each corresponding to tracing out the output of one of the $K+1$ parties. It turns out that the non-trivial part corresponds to tracing out the first party (the transmitter), shown as follows.

\begin{align}
  &\sum_{x\in \mathcal{X}}\mZ''{\footnotesize\left(
  \begin{matrix}
    x\\ \hat{w}_1 \\ \hat{w}_2 \\ \vdots \\ \hat{w}_K
  \end{matrix} \left|  
  \begin{matrix}
    w_1,w_2,\cdots, w_K\\
    y_1  \\
    y_2,s_2^{\tR}\\
    \vdots \\
    y_K,s_K^{\tR}
  \end{matrix}
  \right.
  \right)} \notag \\
  &= \frac{1}{M_1} \sum_{\pi \in \mathbb{C}_{M_1}}
  \sum_{x\in \mathcal{X}}\mZ {\footnotesize \left(
  \begin{matrix}
    x\\ \pi(\hat{w}_1) \\ \hat{w}_2 \\ \vdots \\ \hat{w}_K
  \end{matrix} \left|  
  \begin{matrix}
    \pi(w_1),w_2,\cdots,w_K\\
    y_1,s_1\\
    y_2,s_2^{\tR}\\
    \vdots \\
    y_K,s_K^{\tR}
  \end{matrix}
  \right.
  \right)}\\
  &= \frac{1}{M_1} \sum_{\pi \in \mathbb{C}_{M_1}}
  \mZ_{\mathrmtiny{R_2...R_K}} {\footnotesize \left(
  \begin{matrix}
     \pi(\hat{w}_1) \\ \hat{w}_2 \\ \vdots \\ \hat{w}_K
  \end{matrix} \left|  
  \begin{matrix}
    y_1,s_1\\
    y_2,s_2^{\tR}\\
    \vdots \\
    y_K,s_K^{\tR}
  \end{matrix}
  \right.
  \right)} \label{eq:lemma_extension_1}\\
  &= \frac{1}{M_1} \sum_{\pi \in \mathbb{C}_{M_1}}
  \mZ_{\mathrmtiny{R_2...R_K}} {\footnotesize \left(
  \begin{matrix}
     \hat{w}_2 \\ \vdots \\ \hat{w}_K
  \end{matrix} \left|  
  \begin{matrix}
    y_2, s_2^{\tR}\\
    \vdots \\
    y_K, s_K^{\tR}
  \end{matrix} \label{eq:lemma_extension_2}
  \right.
  \right)}
\end{align}
which does not depend on $(w_1,w_2,\cdots, w_K)$, for all $\hat{w}_k\in \mathcal{M}_k, w_k \in \mathcal{M}_k, y_k \in \mathcal{Y}_k, k\in [K], s_j^{\tR} \in \mathcal{S}_j, j\in \{2,3,\cdots, K\}$. In Step \eqref{eq:lemma_extension_1},  $\mZ_{\mathrmtiny{R_1...R_K}}$ is the marginal distribution of $\mZ$ for User $1$ to User $K$. In Step \eqref{eq:lemma_extension_2}, $\mZ_{\mathrmtiny{R_2...R_K}}$ is the marginal distribution of $\mZ$ for User $2$ to User $K$, and note that $\pi(\hat{w}_1)$ iterates over $[M_1]$ in the summation for every $\hat{w}_1$.

Let us then verify the remaining $K$ NS conditions. Tracing out the output of the second party (User $1$), we have

\begin{align}
  &\sum_{\hat{w}_1}\mZ''{\footnotesize\left(
  \begin{matrix}
    x\\ \hat{w}_1 \\ \hat{w}_2 \\ \vdots \\ \hat{w}_K
  \end{matrix} \left|  
  \begin{matrix}
    w_1,w_2,\cdots, w_K\\
    y_1  \\
    y_2,s_2^{\tR}\\
    \vdots \\
    y_K,s_K^{\tR}
  \end{matrix}
  \right.
  \right)}\\
  &= \frac{1}{M_1} \sum_{\pi \in \mathbb{C}_{M_1}}
  \sum_{\hat{w}_1}\mZ {\footnotesize\left(
  \begin{matrix}
    x\\ \pi(\hat{w}_1) \\ \hat{w}_2 \\ \vdots \\ \hat{w}_K
  \end{matrix} \left|  
  \begin{matrix}
    \pi(w_1),w_2,\cdots, w_K\\
    y_1,s_1\\
    y_2,s_2^{\tR}\\
    \vdots \\
    y_K,s_K^{\tR}
  \end{matrix}
  \right.
  \right)}\\
  &=\frac{1}{M_1} \sum_{\pi \in \mathbb{C}_{M_1}}
  \sum_{\hat{w}_1}\mZ_{\mathrmtiny{TR_2...R_K}} {\footnotesize\left(
  \begin{matrix}
    x   \\ \hat{w}_2 \\ \vdots \\ \hat{w}_K
  \end{matrix} \left|  
  \begin{matrix}
    \pi(w_1),w_2,\cdots, w_K\\
    y_2,s_2^{\tR}\\
    \vdots \\
    y_K,s_K^{\tR}
  \end{matrix}
  \right.
  \right)} \label{eq:lemma_extension_3}
\end{align}%
which does not depend on $y_1$, for all $x\in \mathcal{X}, \hat{w}_j\in \mathcal{M}_j, w_k \in \mathcal{M}_k, y_k \in \mathcal{Y}_k,  s_j^{\tR} \in \mathcal{S}_j, k\in [K], j\in \{2,3,\cdots, K\}$. In Step \eqref{eq:lemma_extension_3}, $\mZ_{\mathrmtiny{TR_2...R_K}}$ is the marginal distribution of $\mZ$ for the transmitter together with User $2$ to User $K$. Similarly one can verify that tracing out $\hat{w}_j$ of $\mZ''$ the distribution does not depend on $(y_j, s_j^{\tR})$, for $j\in \{2,3,\cdots, K\}$.

Finally, by the definition of $\mZ''$ in \eqref{eq:def_Zpp_extention} it follows  that $\eta(\mZ'')=\eta(\mZ')$ and thus equal to $\eta(\mZ)$, which completes the proof of the theorem. \hfill \qed

\bibliographystyle{IEEEtran}
\bibliography{../../bib_file/yy.bib}
\end{document}